\documentclass[11pt]{article}

\usepackage{fullpage}
\usepackage{amsfonts, amssymb, amsmath, amsthm}
\usepackage{latexsym}
\usepackage{url}
\usepackage{color}
\definecolor{DarkGray}{rgb}{0.1,0.1,0.5}
\usepackage[colorlinks=true,breaklinks, linkcolor= DarkGray,citecolor= DarkGray,urlcolor= DarkGray]{hyperref}
\usepackage{graphicx}
\usepackage{subfigure}
\usepackage{cancel}
\usepackage{multirow}
\usepackage{caption}	

\usepackage{import}
\usepackage{ltxtable}

\usepackage{calc}	

\newcommand{\bra}[1]{{\langle#1|}}
\newcommand{\ket}[1]{{|#1\rangle}}
\newcommand{\braket}[2]{{\langle#1|#2\rangle}}
\newcommand{\ketbra}[2]{{\ket{#1}\!\bra{#2}}}
\newcommand{\abs}[1]{{\lvert #1\rvert}}	
\newcommand{\norm}[1]{{\| #1 \|}}

\newcommand{\beq}{\begin{equation}}
\newcommand{\eeq}{\end{equation}}

\def\adjoint{\dagger} 
\def\C {{\bf C}}
\def\N {{\bf N}}

\def\L {{\mathcal L}}
\def\D {{\mathcal D}}
\def\cS {{\ensuremath{\cal S}}}

\DeclareMathOperator{\Span}{\operatorname{Span}}
\DeclareMathOperator{\Range}{\operatorname{Range}}

\newcommand{\identity}{\ensuremath{\boldsymbol{1}}} 

\newcommand{\ADV} {\mathrm{Adv}}
\newcommand{\ADVpm} {\mathrm{Adv}^{\pm}}

\newcommand{\B}{\{0,1\}}

\DeclareMathOperator{\wsizeop}{{\operatorname{wsize}}}		
\newcommand{\wsize}[1]{{\wsizeop({#1})}}

\newcommand{\wsizeS}[2]{{\wsizeop_{#2}({#1})}}

\newcommand{\wsizexS}[3]{{\wsizeop_{#3}({#1},{#2})}}

\DeclareMathOperator{\wsizefop}{{\operatorname{fwsize}}}	\newcommand{\Sf}{S^f}
\newcommand{\wsizef}[1]{{\wsizefop({#1})}}

\newcommand{\wsizefS}[2]{{\wsizefop_{#2}({#1})}}

\newcommand{\wsizefx}[2]{{\wsizefop({#1},{#2})}}
\newcommand{\wsizefxS}[3]{{\wsizefop_{#3}({#1},{#2})}}

\def\Ifree{I_\mathrm{free}}		

\def\jc {\varsigma}	

\DeclareMathOperator{\abst}{\operatorname{abs}}
\def\biadj {B}	

\DeclareMathOperator{\NAND}{\ensuremath{\operatorname{NAND}}}
\DeclareMathOperator{\AND}{\ensuremath{\operatorname{AND}}}
\DeclareMathOperator{\PARITY}{\ensuremath{\operatorname{PARITY}}}
\DeclareMathOperator{\OR}{\ensuremath{\operatorname{OR}}}
\DeclareMathOperator{\MAJ}{{\operatorname{MAJ}_3}}
\DeclareMathOperator{\EQUAL}{{\operatorname{EQUAL}}}

\newcounter{sprows}

\newlength{\spheight}
\newlength{\spraise}

\newcommand{\comment}[1]{\emph{\color{blue}Comment:\color{black} #1}} 
\newlength{\commentslength}
\newcommand{\comments}[1]{
\hspace{-2\parindent}
\addtolength{\commentslength}{-\commentslength}
\addtolength{\commentslength}{\linewidth}
\addtolength{\commentslength}{-\parindent}
\fcolorbox{blue}{white}{\smallskip\begin{minipage}[c]{\commentslength}
\emph{Comments:}\begin{itemize}#1\end{itemize}\end{minipage}}\bigskip
}
\renewcommand{\comment}[1]{}\renewcommand{\comments}[1]{}
\newcommand{\rem}[1]{}

\newtheorem{theorem}{Theorem}[section]
\newtheorem{lemma}[theorem]{Lemma}
\newtheorem{corollary}[theorem]{Corollary}
\newtheorem{claim}[theorem]{Claim}

\newtheorem{definition}[theorem]{Definition}

\newfont{\subsubsecfnt}{ptmri8t at 10pt}
\renewcommand{\subparagraph}[1]{\smallskip{\subsubsecfnt #1.}}

\numberwithin{equation}{section} 

\newcommand{\eqnref}[1]{\hyperref[#1]{{(\ref*{#1})}}}
\newcommand{\thmref}[1]{\hyperref[#1]{{Theorem~\ref*{#1}}}}
\newcommand{\shortthmref}[1]{\hyperref[#1]{{Thm.~\ref*{#1}}}}
\newcommand{\lemref}[1]{\hyperref[#1]{{Lemma~\ref*{#1}}}}
\newcommand{\corref}[1]{\hyperref[#1]{{Corollary~\ref*{#1}}}}
\newcommand{\defref}[1]{\hyperref[#1]{{Definition~\ref*{#1}}}}
\newcommand{\secref}[1]{\hyperref[#1]{{Section~\ref*{#1}}}}
\newcommand{\figref}[1]{\hyperref[#1]{{Figure~\ref*{#1}}}}
\newcommand{\tabref}[1]{\hyperref[#1]{{Table~\ref*{#1}}}}
\newcommand{\remref}[1]{\hyperref[#1]{{Remark~\ref*{#1}}}}
\newcommand{\appref}[1]{\hyperref[#1]{{Appendix~\ref*{#1}}}}
\newcommand{\claimref}[1]{\hyperref[#1]{{Claim~\ref*{#1}}}}
\newcommand{\propref}[1]{\hyperref[#1]{{Proposition~\ref*{#1}}}}
\newcommand{\exampleref}[1]{\hyperref[#1]{{Example~\ref*{#1}}}}
\newcommand{\conjref}[1]{\hyperref[#1]{{Conjecture~\ref*{#1}}}}

\newcommand{\threshold}[2]{{T_{{#1}}^{{#2}}}}  

\allowdisplaybreaks[1]

\begin{document}

\title{Span-program-based quantum algorithm\\ for evaluating unbalanced formulas}
\author{%
Ben W.~Reichardt%
\thanks{School of Computer Science and Institute for Quantum Computing, University of Waterloo.}}
\date{}

\maketitle

\begin{abstract}
The formula-evaluation problem is defined recursively.  A formula's evaluation is the evaluation of a gate, the inputs of which are themselves independent formulas.  Despite this pure recursive structure, the problem is combinatorially difficult for classical computers.  

A quantum algorithm is given to evaluate formulas over any finite boolean gate set.  Provided that the complexities of the input subformulas to any gate differ by at most a constant factor, the algorithm has optimal query complexity.  After efficient preprocessing, it is nearly time optimal.  The algorithm is derived using the span program framework.  It corresponds to the composition of the individual span programs for each gate in the formula.  Thus the algorithm's structure reflects the formula's recursive structure.  
\end{abstract}

\section{Introduction}

A $k$-bit \emph{gate} is a function $f : \{0,1\}^k \rightarrow \{0,1\}$.  A \emph{formula} $\varphi$ over a set of gates $\cS$ is a rooted tree in which each node with $k$ children is associated to a $k$-bit gate from~$\cS$, for $k = 1, 2, \ldots$.  Any such tree with $n$ leaves naturally defines a function $\varphi: \{0,1\}^n \rightarrow \{0,1\}$, by placing the input bits on the leaves in a fixed order and evaluating the gates recursively toward the root.  Such functions are often called \emph{read-once} formulas, as each input bit is associated to one leaf only.  

The formula-evaluation problem is to evaluate a formula $\varphi$ over $\cS$ on an input $x \in \{0,1\}^n$.  The formula is given, but the input string $x$ must be queried one bit at a time.  How many queries to $x$ are needed to compute $\varphi(x)$?  We would like to understand this complexity as a function of $\cS$ and asymptotic properties of $\varphi$.  Roughly, larger gate sets allow $\varphi$ to have less structure, which increases the complexity of evaluating $\varphi$.  Another important factor is often the balancedness of the tree $\varphi$.  Unbalanced formulas often seem to be more difficult to evaluate.  

For applications, the most important gate set consists of all AND and OR gates.  Formulas over this set are known as AND-OR formulas.  Evaluating such a formula solves the decision version of a MIN-MAX tree, also known as a two-player game tree.  Unfortunately, the complexity of evaluating formulas, even over this limited gate set, is unknown, although important special cases have been solved.  The problem over much larger gate sets appears to be combinatorially intractable.  For some formulas, it is known that ``non-directional" algorithms that do not work recursively on the structure of the formula perform better than any recursive procedure.  

In this article, we show that the formula-evaluation problem becomes dramatically simpler when we allow the algorithm to be a bounded-error quantum algorithm, and allow it coherent query access to the input string $x$.  Fix $\cS$ to be any finite set of gates.  We give an optimal quantum algorithm for evaluating ``almost-balanced" formulas over $\cS$.  
The balance condition states that the complexities of the input subformulas to any gate differ by at most a constant factor, where complexity is measured by the general adversary bound $\ADVpm$.  In general, $\ADVpm$ is the value of an exponentially large semi-definite program (SDP).  For a formula $\varphi$ with constant-size gates, though, $\ADVpm(\varphi)$ can be computed efficiently by solving constant-size SDPs for each~gate.  

To place this work in context, some classical and quantum results for evaluating formulas are summarized in \tabref{f:formulastable}.  The stated upper bounds are on query complexity and not time complexity.  However, for the $\OR_n$ and balanced $\AND_2$-$\OR_2$ formulas, the quantum algorithms' running times are only slower by a poly-logarithmic factor.  For the other formulas, the quantum algorithms' running times are slower by a poly-logarithmic factor provided that: 
\begin{enumerate}
\item
A polynomial-time classical preprocessing step, outputting a string $s(\varphi)$, is not charged for. 
\item
The algorithms are allowed unit-cost coherent access to $s(\varphi)$.  
\end{enumerate}

\begin{table}
\begin{center}
\begin{tabular}{|c|c@{~}c|c@{$\!\!\!\!\!\!$}c|}
\hline \hline
  &  \multicolumn{2}{c|}{Randomized, zero-error} & \multicolumn{2}{c|}{Quantum bounded-error} \\
Formula $\varphi$ & \multicolumn{2}{c|}{query complexity $R(\varphi)$} & \multicolumn{2}{c|}{query complexity $Q(\varphi)$} \\
\hline
$\OR_n$ & $n$& & $\Theta(\sqrt n)$&\cite{grover:search, BennettBernsteinBrassardVazirani97upper} \\
\hline
Balanced $\AND_2$-$\OR_2$ & $\Theta(n^\alpha)$&\cite{sw:and-or} & $\Theta(\sqrt n)$&\cite{fgg:and-or, AmbainisChildsReichardtSpalekZhang07andor} \\
\hline
Well-balanced AND-OR & tight recursion&\cite{sw:and-or} & & \\
\hline
Approx.-balanced AND-OR & &
& $\Theta(\sqrt n)$&\cite{AmbainisChildsReichardtSpalekZhang07andor}, $\!$(\shortthmref{t:approxbalancedandor})$\!\!$ \\
\hline
Arbitrary AND-OR & $\Omega(n^{0.51})$&\cite{HeimanWigderson91generalANDOR} & 
\begin{minipage}[l]{0.7in}\begin{center}$\Omega(\sqrt n)$\\$O(\sqrt n \log n)$\end{center}\end{minipage}&\begin{minipage}[l]{0.7in}\begin{center}\cite{BarnumSaks04readonce}\\ \cite{Reichardt09andorfaster}\end{center}\end{minipage} \\
\hline 
Balanced $\MAJ$ ($n = 3^d$) & $\Omega\big((7/3)^d\big)$, $O(2.654^d)$&\cite{JayramKumarSivakumar03majority} & $\Theta(2^d)$&\cite{ReichardtSpalek08spanprogram} \\
\hline
Balanced over $\cS$
 & 
 & & $\Theta(\ADVpm(\varphi))$&\cite{Reichardt09spanprogram} \\
\hline
Almost-balanced over $\cS$
 & 
 & & $\Theta(\ADVpm(\varphi))$&(\shortthmref{t:unbalancedformulaevaluation}) \\
\hline \hline
\end{tabular}
\end{center}
\caption{Comparison of some classical and quantum query complexity results for formula evaluation.  Here $\cS$ is any fixed, finite gate set, and the exponent $\alpha$ is given by $\alpha = \log_2 (\frac{1 + \sqrt{33}}{4}) \approx 0.753$.  Under certain assumptions, the algorithms' running times are only poly-logarithmically slower.} 
\label{f:formulastable}
\end{table}

Our algorithm is based on the framework relating span programs and quantum algorithms from~\cite{Reichardt09spanprogram}.  Previous work has used span programs to develop quantum algorithms for evaluating formulas~\cite{ReichardtSpalek08spanprogram}.  Using this and the observation that the optimal span program witness size for a boolean function $f$ equals the general adversary bound $\ADVpm(f)$, Ref.~\cite{Reichardt09spanprogram} gives an optimal quantum algorithm for evaluating ``adversary-balanced" formulas over an arbitrary finite gate set.  The balance condition is that each gate's input subformulas have 
equal general adversary bounds.  

In order to relax this strict balance requirement, we must maintain better control in the recursive analysis.  To help do so, we define a new span program complexity measure, the ``full witness size."  This complexity measure has implications for developing time- and query-efficient quantum algorithms based on span programs.  Essentially, using a second result from~\cite{Reichardt09spanprogram}, that properties of eigenvalue-zero eigenvectors of certain bipartite graphs imply ``effective" spectral gaps around zero, it allows quantum algorithms to be based on span programs with free inputs.  This can simplify the implementation of a quantum walk on the corresponding graph.  

Besides allowing a relaxed balance requirement, our approach has the additional advantage of making the constants hidden in the big-$O$ notation more explicit.  The formula-evaluation quantum algorithms in~\cite{ReichardtSpalek08spanprogram, Reichardt09spanprogram} evaluate certain formulas $\varphi$ using $O\big(\ADVpm(\varphi)\big)$ queries, where the hidden constant depends on the gates in $\cS$ in a complicated manner.  It is not known how to upper-bound the hidden constant in terms of, say, the maximum fan-in $k$ of a gate in $\cS$.  In contrast, the approach we follow here allows bounding this constant by an exponential in $k$.  

It is known that the general adversary bound is a nearly tight lower bound on quantum query complexity for \emph{any} boolean function~\cite{Reichardt09spanprogram}, including in particular boolean formulas.  However, this comes with no guarantees on time complexity.  The main contribution of this paper is to give a nearly time-optimal algorithm for formula evaluation.  The algorithm is also tight for query complexity, removing the extra logarithmic factor from the bound in~\cite{Reichardt09spanprogram}.  

Additionally, we apply the same technique to study AND-OR formulas.  For this special case, special properties of span programs for AND and for OR gates allow the almost-balance condition to be significantly weakened.  Ambainis et al.~\cite{AmbainisChildsReichardtSpalekZhang07andor} have studied this case previously.  By applying the span program framework, we identify a slight weakness in their analysis.  Tightening the analysis extends the algorithm's applicability to a broader class of AND-OR formulas.  

A companion paper~\cite{Reichardt09andorfaster} applies the span program framework to the problem of evaluating \emph{arbitrary} AND-OR formulas.  By studying the full witness size for span programs constructed using a novel composition method, it gives an $O(\sqrt n \log n)$-query quantum algorithm to evaluate a formula of size $n$, for which the time complexity is poly-logarithmically worse after preprocessing.  This nearly matches the $\Omega(\sqrt n)$ lower bound, and improves a $\sqrt n 2^{O(\sqrt {\log n})}$-query quantum algorithm from~\cite{AmbainisChildsReichardtSpalekZhang07andor}.  Ref.~\cite{Reichardt09andorfaster} shares the broader motivation of this paper, to study span program properties and design techniques that lead to time-efficient quantum algorithms.  

Sections~\ref{s:classicalhistory} and~\ref{s:quantumhistory} below give further background on the formula-evaluation problem, for classical and quantum algorithms.  \secref{s:quantumresults} precisely states our main theorem, the proof of which is given in \secref{s:approxbalancedproof} after some background on span programs.  The theorem for approximately balanced AND-OR formulas is stated in \secref{s:quantumANDORapproxbalanceresults}, and proved in \secref{s:approxbalancedandor}.  An appendix revisits the proof from~\cite{AmbainisChildsReichardtSpalekZhang07andor} to prove our extension directly, without using the span program framework.

\subsection{History of the formula-evaluation problem for classical algorithms} \label{s:classicalhistory}

For a function $f : \{0,1\}^n \rightarrow \{0,1\}$, let $D(f)$ be the least number of input bit queries sufficient to evaluate $f$ on any input with zero error.  $D(f)$ is known as the deterministic decision-tree complexity of $f$, or the deterministic query complexity of $f$.  Let the randomized decision-tree complexity of $f$, $R(f) \leq D(f)$, be the least \emph{expected} number of queries required to evaluate $f$ with zero error (i.e., by a Las Vegas randomized algorithm).  Let the Monte Carlo decision-tree complexity, $R_2(f) = O\big(R(f)\big)$, be the least number of queries required to evaluate $f$ with error probability at most $1/3$ (i.e., by a Monte Carlo randomized algorithm).  

Classically, formulas over the gate set $\cS = \{ \NAND_k : k \in \N \}$ have been studied most extensively, where $\NAND_k(x_1, \ldots, x_k) = 1 - \prod_{j=1}^k x_j$.  By De Morgan's rules, any formula over $\NAND$ gates can also be written as a formula in which the gates at an even distance from the formula's root are $\AND$ gates and those an odd distance away are $\OR$ gates, with some inputs or the output possibly complemented.  Thus formulas over $\cS$ are also known as AND-OR formulas.  	

For any AND-OR formula $\varphi$ of \emph{size} $n$, i.e., on $n$ inputs, $D(\varphi) = n$.  However, randomization gives a strict advantage; $R(\varphi)$ and $R_2(\varphi)$ can be strictly smaller.  Indeed, let $\varphi_d$ be the complete, binary AND-OR formula of depth $d$, corresponding to the tree in which each internal vertex has two children and every leaf is at distance $d$ from the root, with alternating levels of AND and OR gates.  Its size is $n = 2^d$.  
Snir~\cite{snir:dec} has given a randomized algorithm for evaluating $\varphi_d$ using in expectation $O(n^\alpha)$ queries, where $\alpha = \log_2 (\frac{1+\sqrt{33}}{4}) \approx 0.753$~\cite{sw:and-or}.  This algorithm, known as randomized alpha-beta pruning, evaluates a random subformula recursively, and only evaluates the second subformula if necessary.  Saks and Wigderson~\cite{sw:and-or} have given a matching lower bound on $R(\varphi_d)$, which Santha has extended to hold for Monte Carlo algorithms, $R_2(\varphi_d) = \Omega(n^\alpha)$~\cite{santha:and-or}.  

Thus the query complexities have been characterized for the complete, binary AND-OR formulas.  In fact, the tight characterization works for a larger class of formulas, called ``well balanced" formulas by~\cite{santha:and-or}.  This class includes, for example, alternating $\AND_2$-$\OR_2$ formulas where for some $d$ every leaf is at depth $d$ or $d-1$, Fibonacci trees and binomial trees~\cite{sw:and-or}.  It also includes skew trees, for which the depth is the maximal $n-1$.  

For arbitrary AND-OR formulas, on the other hand, little is known.  It has been conjectured that complete, binary AND-OR formulas are the easiest to evaluate, and that in particular $R(\varphi) = \Omega(n^\alpha)$ for any size-$n$ AND-OR formula $\varphi$~\cite{sw:and-or}.  However, the best general lower bound is $R(\varphi) = \Omega(n^{0.51})$, due to Heiman and Wigderson~\cite{HeimanWigderson91generalANDOR}.  Ref.~\cite{HeimanWigderson91generalANDOR} also extends the result of \cite{sw:and-or} to allow for AND and OR gates with fan-in more than two.  

It is perhaps not surprising that formulas over most other gate sets $\cS$ are even less well understood.  For example, Boppana has asked the complexity of evaluating the complete ternary majority ($\MAJ$) formula of depth $d$~\cite{sw:and-or}, and the best published bounds on its query complexity are $\Omega\big((7/3)^d\big)$ and $O\big((2.6537\ldots)^d\big)$~\cite{JayramKumarSivakumar03majority}.  In particular, the na{\" i}ve, ``directional," generalization of the randomized alpha-beta pruning algorithm is to evaluate recursively two random immediate subformulas and, if they disagree, then also the third.  This algorithm uses $O\big((8/3)^d\big)$ expected queries, and is suboptimal.  This suggests that the complete $\MAJ$ formulas are significantly different from the complete AND-OR formulas.  

\newcommand{\TCzero}{\mathsf{TC^0}}
\newcommand{\NCone}{\mathsf{NC^1}}

Heiman, Newman and Wigderson have considered read-once threshold formulas in an attempt to separate the complexity classes $\TCzero$ from $\NCone$~\cite{HeimanNewmanWigderson90threshold}.  That is, they allow the gate set to be the set of Hamming-weight threshold gates $\{ \threshold{m}{k} : m, k \in \N \}$ defined by $\threshold{m}{k}: \{0, 1\}^k \rightarrow \{0,1\}$, $\threshold{m}{k}(x) = 1$ if and only if the Hamming weight of $x$ is at least $m$.  AND, OR and majority gates are all special cases of threshold gates.  Heiman et al.\ prove that $R(\varphi) \geq n/2^d$ for $\varphi$ a threshold formula of depth $d$, and in fact their proof extends to gate sets in which every gate ``contains a flip"~\cite{HeimanNewmanWigderson90threshold}.  This implies that a large depth is necessary for the randomized complexity to be much lower than the deterministic complexity.  

Of course there are some trivial gate sets for which the query complexity is fully understood, for example, the set of parity gates.  Overall, though, there are many more open problems than results.  Despite its structure, formula evaluation appears to be combinatorially complicated.  However, there is another approach, to try to leverage the power of quantum computers.  Surprisingly, the formula-evaluation problem simplifies considerably in this different model of computation.

\subsection{History of the formula-evaluation problem for quantum algorithms} \label{s:quantumhistory}

In the quantum query model, the input bits can be queried coherently.  That is, the quantum algorithm is allowed unit-cost access to the unitary operator $O_x$, called the input oracle, defined by 
\beq \label{e:query}
O_x: \, \ket \varphi \otimes \ket j \otimes \ket b \mapsto 
\ket \varphi \otimes \ket j \otimes \ket{b \oplus x_j}
 \enspace .
\eeq
Here $\ket \varphi$ is an arbitrary pure state, $\{\ket j : j = 1,2,\ldots,n\}$ is an orthonormal basis for $\C^{n}$, $\{\ket b : b = 0, 1\}$ is an orthonormal basis for $\C^2$, and $\oplus$ denotes addition mod two.  $O_x$ can be implemented efficiently on a quantum computer given a classical circuit that computes the function $j \mapsto x_j$~\cite{NielsenChuang00}.  For a function $f : \{0,1\}^n \rightarrow \{0,1\}$, let $Q(f)$ be the number of input queries required to evaluate $f$ with error probability at most $1/3$.  It is immediate that $Q(f) \leq R_2(f)$.  

Research on the formula-evaluation problem in the quantum model began with the $n$-bit OR function, $\OR_n$.  Grover gave a quantum algorithm for evaluating $\OR_n$ with bounded one-sided error using $O(\sqrt n)$ oracle queries and $O(\sqrt n \log \log n)$ time~\cite{grover:search, grover:search-time}.  In the classical case, on the other hand, it is obvious that $R_2(\OR_n)$, $R(\OR_n)$ and $D(\OR_n)$ are all~$\Theta(n)$.  

Grover's algorithm can be applied recursively to speed up the evaluation of more general AND-OR formulas.  Call a formula \emph{layered} if the gates at the same depth are the same.  Buhrman, Cleve and Wigderson show that a layered, depth-$d$, size-$n$ AND-OR formula can be evaluated using $O(\sqrt n \log^{d-1} n)$ queries~\cite{BuhrmanCleveWigderson98}.  The logarithmic factors come from using repetition at each level to reduce the error probability from a constant to be polynomially small.  

H\o yer, Mosca and de Wolf \cite{HoyerMoscaDeWolf03berror-search} consider the case of a unitary input oracle $\tilde O_x$ that maps 
\beq
\tilde O_x: \, \ket \varphi \otimes \ket j \otimes \ket b \otimes \ket 0
\mapsto 
\ket \varphi \otimes \ket j \otimes \big( \ket{b \oplus x_j} \otimes \ket{\psi_{x,j,x_j}} + \ket{b \oplus \overline x_j} \otimes \ket{\psi_{x,j,\overline x_j}} \big)
 \enspace ,
\eeq
where $\ket{\psi_{x,j,x_j}}$, $\ket{\psi_{x,j,\overline x_j}}$ are pure states with $\norm{\ket{\psi_{x,j,x_j}}}^2 \geq 2/3$.  Such an oracle can be implemented when the function $j \mapsto x_j$ is computed by a bounded-error, randomized subroutine.  H\o yer et al.\ allow access to $\tilde O_x$ and $\tilde O_x^{-1}$, both at unit cost, and show that $\OR_n$ can still be evaluated using $O(\sqrt n)$ queries.  This \emph{robustness} result implies that the $\log n$ steps of repetition used by~\cite{BuhrmanCleveWigderson98} are not necessary, and a depth-$d$ layered AND-OR formula can be computed in $O(\sqrt n \, c^{d-1})$ queries, for some constant $c > 1000$.  If the depth is constant, this gives an $O(\sqrt n)$-query quantum algorithm, but the result is not useful for the complete, binary AND-OR formula, for which $d = \log_2 n$.  

In 2007, Farhi, Goldstone and Gutmann presented a quantum algorithm for evaluating complete, binary AND-OR formulas~\cite{fgg:and-or}.  Their breakthrough algorithm is not based on iterating Grover's algorithm in any way, but instead runs a quantum walk---analogous to a classical random walk---on a graph based on the formula.  The algorithm runs in time $O(\sqrt n)$ in a certain continuous-time query model.  

Ambainis et al.\ discretized the~\cite{fgg:and-or} algorithm by reinterpreting a correspondence between (discrete-time) random and quantum walks due to Szegedy~\cite{Szegedy04walkfocs} as a correspondence between continuous-time and discrete-time quantum walks~\cite{AmbainisChildsReichardtSpalekZhang07andor}.  Applying this correspondence to quantum walks on certain \emph{weighted} graphs, they gave an $O(\sqrt n)$-query quantum algorithm for evaluating ``approximately balanced" AND-OR formulas.  For example, $\MAJ(x_1, x_2, x_3) = (x_1 \wedge x_2) \vee \big((x_1 \vee x_2) \wedge x_3\big)$, so there is a size-$5^d$ AND-OR formula that computes $\MAJ^d$ the complete ternary majority formula of depth $d$.  Since the formula is approximately balanced, $Q(\MAJ^d) = O(\sqrt 5^d)$, better than the $\Omega\big((7/3)^d\big)$ classical lower bound.  

The~\cite{AmbainisChildsReichardtSpalekZhang07andor} algorithm also applies to arbitrary AND-OR formulas.  If $\varphi$ has size $n$ and depth $d$, then the algorithm, applied directly, evaluates $\varphi$ using $O(\sqrt n \, d)$ queries.\footnote{Actually, \cite[Sec.~7]{AmbainisChildsReichardtSpalekZhang07andor} only shows a bound of $O(\sqrt n \, d^{3/2})$ queries, but this can be improved to $O(\sqrt n \, d )$ using the bounds on $\sigma_\pm(\varphi)$ below \cite[Def.~1]{AmbainisChildsReichardtSpalekZhang07andor}.}  This can be as bad as $O(n^{3/2})$ if the depth is $d = n$.  However, Bshouty, Cleve and Eberly have given a formula rebalancing procedure that takes AND-OR formula $\varphi$ as input and outputs an equivalent AND-OR formula $\varphi'$ with depth $d' = 2^{O(\sqrt{\log n})}$ and size $n' = n \, 2^{O(\sqrt {\log n})}$~\cite{bce:size-depth, bb:size-depth}.  The formula $\varphi'$ can then be evaluated using $O(\sqrt {n'} \, d' ) = \sqrt n \, 2^{O(\sqrt{\log n})}$ queries.  

\medskip

Our understanding of \emph{lower} bounds for the formula-evaluation problem progressed in parallel to this progress on quantum algorithms.  There are essentially two techniques, the \emph{polynomial} and \emph{adversary} methods, for lower-bounding quantum query complexity.  
\begin{itemize}
\item
The polynomial method, introduced in the quantum setting by Beals et al.~\cite{BealsBuhrmanCleveMoscaWolf98}, 
is based on the observation that after making $q$ oracle $O_x$ queries, the probability of any measurement result is a polynomial of degree at most $2q$ in the variables~$x_j$.  
\item
Ambainis generalized the classical hybrid argument, to consider the system's entanglement when run on a superposition of inputs~\cite{Ambainis00adversary}.  A number of variants of Ambainis's bound were soon discovered, including weighted versions~\cite{HoyerNeerbekShi02adv, BarnumSaks04readonce, Ambainis06polynomial, Zhang05adv}, a spectral version~\cite{BarnumSaksSzegedy03adv}, and a version based on Kolmogorov complexity~\cite{LaplanteMagniez04kolmogorov}.  These variants can be asymptotically stronger than Ambainis's original unweighted bound, but are equivalent to each other~\cite{SpalekSzegedy04advequivalent}.  We therefore term it simply ``the adversary bound," denoted by $\ADV$.  
\end{itemize}

The adversary bound is well-suited for lower-bounding the quantum query complexity for evaluating formulas.   For example, Barnum and Saks proved that for any size-$n$ AND-OR formula $\varphi$, $\ADV(\varphi) = \sqrt n$, implying the lower bound $Q(\varphi) = \Omega(\sqrt n)$~\cite{BarnumSaks04readonce}.  Thus the \cite{AmbainisChildsReichardtSpalekZhang07andor} algorithm is optimal for approximately balanced AND-OR formulas, and is nearly optimal for arbitrary AND-OR formulas.  This is a considerably more complete solution than is known classically.  

\medskip

It is then natural to consider formulas over larger gate sets.  The adversary bound continues to work well, because it transforms nicely under function composition: 

\begin{theorem}[Adversary bound composition {\cite{Ambainis06polynomial, LaplanteLeeSzegedy06adversary, HoyerLeeSpalek05compose}}] \label{t:weaknonnegativeadversarycomposition}
Let $f : \{0,1\}^k \rightarrow \{0,1\}$ and let $f_j : \{0,1\}^{m_j} \rightarrow \{0,1\}$ for $j = 1, 2, \ldots, k$.  Define $g : \{0,1\}^{m_1} \times \cdots \times \{0,1\}^{m_k} \rightarrow \{0,1\}$ by $g(x) = f\big(f_1(x_1), \ldots, f_k(x_k)\big)$.  Let $s = (\ADV(f_1), \ldots, \ADV(f_k))$.  Then 
\begin{align} \label{e:weaknonnegativeadversarycomposition}
\ADV(g) &= \ADV_s(f)
 \enspace .
\end{align}
\end{theorem}

See \defref{t:adversarydef} for the definition of the adversary bound with ``costs," $\ADV_s$.  The $\ADV$ bound equals $\ADV_s$ with uniform, unit costs $s = \vec 1$.  For a function $f$, $\ADV(f)$ can be computed using a semi-definite program in time polynomial in the size of $f$'s truth table.  Therefore, \thmref{t:weaknonnegativeadversarycomposition} gives a polynomial-time procedure for computing the adversary bound for a formula $\varphi$ over an arbitrary finite gate set: compute the bounds for subformulas, moving from the leaves toward the root.  At an internal node $f$, having computed the adversary bounds for the input subformulas $f_1, \ldots, f_k$, Eq.~\eqnref{e:weaknonnegativeadversarycomposition} says that the adversary bound for $g$, the subformula rooted at $f$, equals the adversary bound for the \emph{gate} $f$ with certain costs.  Computing this requires $2^{O(k)}$ time, which is a constant if $k = O(1)$.  For example, if $f$ is an $\OR_k$ or $\AND_k$ gate, then $\ADV_{(s_1, \ldots, s_k)}(f) = \sqrt{\sum_j s_j^2}$, from which follows immediately the~\cite{BarnumSaks04readonce} result $\ADV(\varphi) = \sqrt n$ for a size-$n$ AND-OR formula $\varphi$.  

A special case of \thmref{t:weaknonnegativeadversarycomposition} is when the functions $f_j$ all have equal adversary bounds, so $\ADV(g) = \ADV(f) \ADV(f_1)$.  In particular, for a function $f : \{0,1\}^k \rightarrow \{0,1\}$ and a natural number $d \in \N$, let $f^d : \{0,1\}^{k^d} \rightarrow \{0,1\}$ denote the complete, depth-$d$ formula over $f$.  That is, $f^1 = f$ and $f^d(x) = f \big(f^{d-1}(x_1,\ldots, x_{k^{d-1}}), \ldots, f^{d-1}(x_{k^d-k^{d-1}+1}, \ldots, x_{k^d}) \big)$ for $d > 1$.  Then we obtain: 

\begin{corollary} \label{t:nonnegativeadversarycompleteformula}
For any function $f : \{0,1\}^k \rightarrow \{0,1\}$, 
\beq
\ADV(f^d) = \ADV(f)^d
 \enspace .
\eeq
\end{corollary}

In particular, Ambainis defined a boolean function $f : \{0,1\}^4 \rightarrow \{0,1\}$ that can be represented exactly by a polynomial of degee two, but for which $\ADV(f) = 5/2$~\cite{Ambainis06polynomial}.  Thus $f^d$ can be represented exactly by a polynomial of degree $2^d$, but by \corref{t:nonnegativeadversarycompleteformula}, $\ADV(f^d) = (5/2)^d$.  
For this function, the adversary bound is strictly stronger than any bound obtainable using the polynomial method.  Many similar examples are given in~\cite{HoyerLeeSpalek07negativeadvurl}.  However, for other functions, the adversary bound is asymptotically worse than the polynomial method~\cite{SpalekSzegedy04advequivalent, AaronsonShi04collisioned, Ambainis05polynomialmethod}.  

In 2007, though, H{\o}yer et al.\ discovered a strict generalization of $\ADV$ that also lower-bounds quantum query complexity~\cite{HoyerLeeSpalek07negativeadv}.  We call this new bound the \emph{general} adversary bound, or $\ADVpm$.  For example, for Ambainis's four-bit function $f$, $\ADVpm(f) \geq 2.51$~\cite{HoyerLeeSpalek07negativeadvurl}.  Like the adversary bound, $\ADVpm_s(f)$ can be computed in time polynomial in the size of $f$'s truth table, and also composes nicely: 

\begin{theorem}[{\cite{HoyerLeeSpalek07negativeadv, Reichardt09spanprogram}}] \label{t:adversarycomposition}
Under the conditions of \thmref{t:weaknonnegativeadversarycomposition}, 
\begin{align}
\ADVpm(g) &= \ADVpm_s(f)
 \enspace .
\end{align}
In particular, if $\ADVpm(f_1) = \cdots = \ADVpm(f_k)$, then $\ADVpm(g) = \ADVpm(f) \, \ADVpm(f_1)$.  
\end{theorem}

Define a formula $\varphi$ to be \emph{adversary balanced} if at each internal node, the general adversary bounds of the input subformulas are equal.  In particular, by \thmref{t:adversarycomposition} this implies that $\ADVpm(\varphi)$ is equal to the product of the general adversary bounds of the gates along any path from the root to a leaf.  Complete, {layered} formulas are an example of adversary-balanced formulas.  

\medskip

Returning to upper bounds, Reichardt and {\v S}palek~\cite{ReichardtSpalek08spanprogram} generalized the algorithmic approach started by~\cite{fgg:and-or}.  They gave an optimal quantum algorithm for evaluating adversary-balanced formulas over a considerably extended gate set, including in particular all functions $\{0,1\}^k \rightarrow \{0,1\}$ for $k \leq 3$, $69$ inequivalent four-bit functions, and the gates $\AND_k$, $\OR_k$, $\PARITY_k$ and $\EQUAL_k$, for $k = O(1)$.  For example, $Q(\MAJ^d) = \Theta(2^d)$.  

The~\cite{ReichardtSpalek08spanprogram} result follows from a framework for developing formula-evaluation quantum algorithms based on \emph{span programs}.  A span program, introduced by Karchmer and Wigderson~\cite{KarchmerWigderson93span}, is a certain linear-algebraic way of defining a function, which corresponds closely to eigenvalue-zero eigenvectors of certain bipartite graphs.  \cite{ReichardtSpalek08spanprogram} derived a quantum algorithm for evaluating certain concatenated span programs, with a query complexity upper-bounded by the span program \emph{witness size}, denoted wsize.  In particular, a special case of~\cite[Theorem~4.7]{ReichardtSpalek08spanprogram} is: 

\begin{theorem}[\cite{ReichardtSpalek08spanprogram}] \label{t:spanprogramalgorithm}
Fix a function $f : \{0,1\}^k \rightarrow \{0,1\}$.  
If span program $P$ computes $f$, then 
\beq
Q(f^d) = O\big(\wsize{P}^d\big)
 \enspace .
\eeq
\end{theorem}

From \thmref{t:adversarycomposition}, this result is optimal if $\wsize{P} = \ADVpm(f)$.  The question therefore becomes how to find optimal span programs.  Using an ad hoc search, \cite{ReichardtSpalek08spanprogram} found optimal span programs for a variety of functions with $\ADVpm = \ADV$.  Further work automated the search, by giving a semi-definite program (SDP) for the optimal span program witness size for any given function~\cite{Reichardt09spanprogram}.  
Remarkably, the SDP's value always equals the general adversary bound: 

\begin{theorem}[\cite{Reichardt09spanprogram}] \label{t:spanprogramSDPintro}
For any function $f : \{0,1\}^n \rightarrow \{0,1\}$,
\beq
\inf_{P} \wsize P  = \ADVpm(f)
 \enspace ,
\eeq
where the infimum is over span programs $P$ computing $f$.  Moreover, this infimum is achieved.  
\end{theorem}

This result greatly extends the gate set over which the formula-evaluation algorithm of~\cite{ReichardtSpalek08spanprogram} works optimally.  For example, combined with \thmref{t:spanprogramalgorithm}, it implies that $\lim_{d \rightarrow \infty} Q(f^d)^{1/d} = \ADVpm(f)$ for every boolean function $f$.  More generally, \thmref{t:spanprogramSDPintro} allows the~\cite{ReichardtSpalek08spanprogram} algorithm to be run on formulas over any finite gate set $\cS$.  A factor is lost that depends on the gates in $\cS$, but it will be a constant for $\cS$ finite.  Combining \thmref{t:spanprogramSDPintro} with \cite[Theorem~4.7]{ReichardtSpalek08spanprogram} gives: 

\begin{theorem}[\cite{Reichardt09spanprogram}] \label{t:formulaevaluation}
Let $\cS$ be a finite set of gates.  
Then there exists a quantum algorithm that evaluates an adversary-balanced formula $\varphi$ over $\cS$ using $O\big(\ADVpm(\varphi)\big)$ input queries.  After efficient classical preprocessing independent of the input $x$, and assuming unit-time coherent access to the preprocessed classical string, the running time of the algorithm is $\ADVpm(\varphi) \big(\log \ADVpm(\varphi)\big)^{O(1)}$.  
\end{theorem}

In the discussion so far, we have for simplicity focused on query complexity.  The query complexity is an information-theoretic quantity that does not charge for operations independent of the input string, even though these operations may require many elementary gates to implement.  For practical algorithms, it is important to be able to bound the algorithm's \emph{running time}, which counts the cost of implementing the input-independent operations.  \thmref{t:formulaevaluation} puts an optimal bound on the query complexity, and also puts a nearly optimal bound on the algorithm's time complexity.  In fact, all of the query-optimal algorithms so far discussed are also nearly time optimal.  

In general, though, an upper bound on the query complexity does not imply an upper bound on the time complexity.  Ref.~\cite{Reichardt09spanprogram} also generalized the span program framework of~\cite{ReichardtSpalek08spanprogram} to apply to quantum algorithms not based on formulas.  The main result of~\cite{Reichardt09spanprogram} is: 

\begin{theorem}[\cite{Reichardt09spanprogram}] \label{t:querycomplexitytightnonbinary}
For any function $f : \D \rightarrow \{1, 2, \ldots, m\}$, with $\D \subseteq \{0,1\}^n$, $Q(f)$ satisfies
\begin{align}
Q(f) = \Omega(\ADVpm(f)) 
\quad \text{and} \quad
Q(f) = O\bigg(\ADVpm(f) \, \frac{\log \ADVpm(f)}{\log \log \ADVpm(f)} \log(m) \log \log m \bigg)
 \enspace .
\end{align}
\end{theorem}

\thmref{t:querycomplexitytightnonbinary} in particular allows us to compute the query complexity of formulas, up to the logarithmic factor.  It does \emph{not} give any guarantees on running time.  However, the analysis required to prove \thmref{t:querycomplexitytightnonbinary} also leads to significantly simpler proofs of \thmref{t:formulaevaluation} and the AND-OR formula results of~\cite{AmbainisChildsReichardtSpalekZhang07andor, fgg:and-or}.  Moreover, we will see that it allows the formula-evaluation algorithms to be extended to formulas that are not adversary balanced.

\subsection{Quantum algorithm for evaluating almost-balanced formulas} \label{s:quantumresults}

\def\pathsum #1{{\sigma_+({#1})}}
\def\pathinvsum #1{{\sigma_-(#1)}}
\def\path {\xi}

We give a formula-evaluation algorithm that is both query-optimal, without a logarithmic overhead, and, after an efficient preprocessing step, nearly time optimal.  Define almost balance as follows: 

\begin{definition} \label{t:approxbalancedef}
Consider a formula $\varphi$ over a gate set $\cS$.  For a vertex $v$ in the corresponding tree, let $\varphi_v$ denote the subformula of $\varphi$ rooted at $v$, and, if $v$ is an internal vertex, let $g_v$ be the corresponding gate.  The formula $\varphi$ is \emph{$\beta$-balanced} if for every vertex $v$, with children $c_1, c_2, \ldots, c_k$, 
\beq \label{e:approxbalancedef}
\frac{\max_{j} \ADVpm(\varphi_{c_j})}{\min_{j} \ADVpm(\varphi_{c_j})} \leq \beta
 \enspace .
\eeq
(If $c_j$ is a leaf, $\ADVpm(\varphi_{c_j}) = 1$.)  Formula $\varphi$ is \emph{almost balanced} if it is $\beta$-balanced for some $\beta = O(1)$.  
\end{definition}

In particular, an adversary-balanced formula is $1$-balanced.  We will show: 

\begin{theorem} \label{t:unbalancedformulaevaluation}
Let $\cS$ be a fixed, finite set of gates.  Then there exists a quantum algorithm that evaluates an almost-balanced formula $\varphi$ over $\cS$ using $O\big(\ADVpm(\varphi)\big)$ input queries.  After polynomial-time classical preprocessing independent of the input, and assuming unit-time coherent access to the preprocessed string, the running time of the algorithm is $\ADVpm(\varphi) \big(\log \ADVpm(\varphi)\big)^{O(1)}$.  
\end{theorem}

\thmref{t:unbalancedformulaevaluation} is significantly stronger than \thmref{t:formulaevaluation}, which requires exact balance.  There are important classes of exactly balanced formulas, such as complete, layered formulas.  In fact, it is sufficient that the multiset of gates along the simple path from the root to a leaf not depend on the leaf.  Moreover, sometimes different gates have the same $\ADVpm$ bound; see~\cite{HoyerLeeSpalek07negativeadvurl} for examples.  Even still, exact adversary balance is a very strict condition.  

The proof of \thmref{t:unbalancedformulaevaluation} is based on the span program framework developed in Ref.~\cite{Reichardt09spanprogram}.  In particular, \cite[Theorem~9.1]{Reichardt09spanprogram} gives \emph{two} quantum algorithms for evaluating span programs.  The first algorithm is based on a discrete-time simulation of a continuous-time quantum walk.  It applies to arbitrary span programs, and is used, in combination with \thmref{t:spanprogramSDPintro}, to prove \thmref{t:querycomplexitytightnonbinary}.  However, the simulation incurs a logarithmic query overhead and potentially worse time complexity overhead, so this algorithm is not suitable for proving \thmref{t:unbalancedformulaevaluation}.  

The second algorithm in~\cite{Reichardt09spanprogram} is based directly on a discrete-time quantum walk, similar to previous optimal formula-evaluation algorithms~\cite{AmbainisChildsReichardtSpalekZhang07andor, ReichardtSpalek08spanprogram}.  However, this algorithm does not apply to an arbitrary span program.  A bound is needed on the operator norm of the entry-wise absolute value of the weighted adjacency matrix for a corresponding graph.  Further graph sparsity conditions are needed for the algorithm to be time efficient (see \thmref{t:generalspanprogramalgorithmnonblackbox}).  

Unfortunately, the span program from \thmref{t:spanprogramSDPintro} will not generally satisfy these conditions.  \thmref{t:spanprogramSDPintro} gives a \emph{canonical} span program (\cite[Def.~5.1]{Reichardt09spanprogram}).  Even for a simple formula, the optimal canonical span program will typically correspond to a dense graph with large norm.  

An example should clarify the problem.  Consider the AND-OR formula $\psi(x) = \big( [ (x_1 \wedge x_2) \vee x_3 ] \wedge x_4 \big) \vee \big( x_5 \wedge [x_6 \vee x_7] \big)$, and consider the two graphs in \figref{f:graphexamples}.  For an input $x \in \{0,1\}^7$, modify the graphs by attaching dangling edges to every vertex $j$ for which $x_j = 0$.  Observe then that each graph has an eigenvalue-zero eigenvector supported on vertex $0$---called a \emph{witness}---if and only if $\psi(x) = 1$.  The graphs correspond to different span programs computing $\psi$, and the quantum algorithm works essentially by running a quantum walk starting at vertex $0$ in order to detect the witness.  The graph on the left is a significantly simplified version of a canonical span program for $\psi $, and its density still makes it difficult to implement the quantum walk.  

\begin{figure}
\centering
\begin{tabular}{c@{$\quad$}c}
\subfigure[]{\label{f:tensorproductgraphexample}\includegraphics[scale=.75]{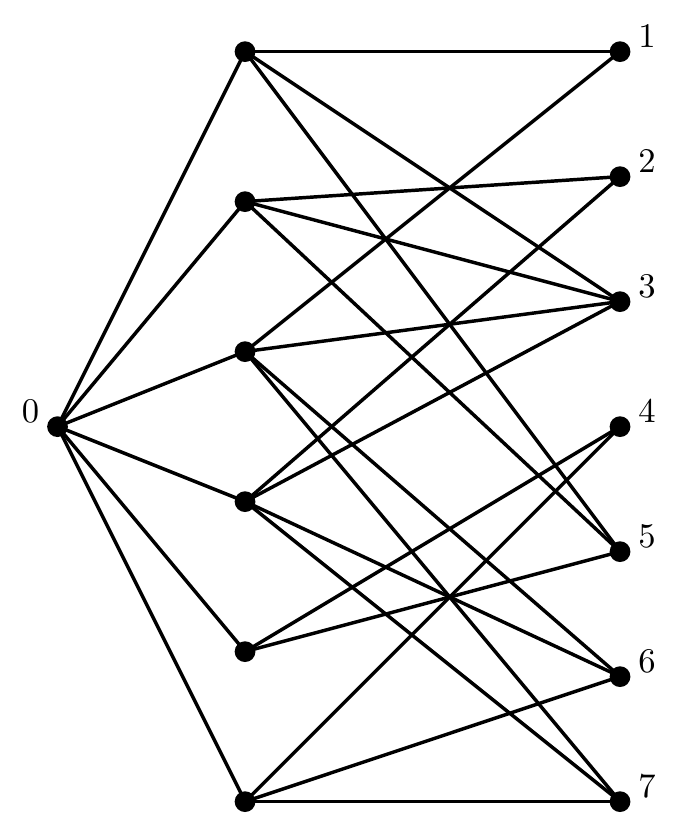}}
&
\subfigure[]{\label{f:directsumgraphexample}\raisebox{1.40cm}{\includegraphics[scale=.75]{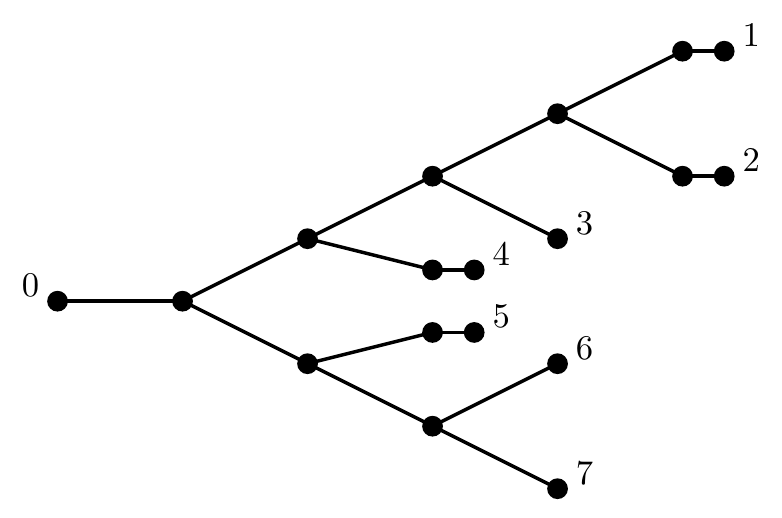}}}
\end{tabular}
\caption{Graphs corresponding to two span programs both computing the same function.} \label{f:graphexamples}
\end{figure}

We will be guided by the second, simpler graph.  Instead of applying \thmref{t:spanprogramSDPintro} to $\varphi$ as a whole, we apply it separately to every gate in the formula.  We then compose these span programs, one per gate, according to the formula, using \emph{direct-sum composition} (\defref{t:directsumcomposedef}).  In terms of graphs, direct-sum composition attaches the output vertex of one span program's graph to an input vertex of the next~\cite{ReichardtSpalek08spanprogram}.  This leads to a graph whose structure somewhat follows the structure of the formula $\varphi$, as the graph in \figref{f:directsumgraphexample} follows the structure of $\psi$.  (However, the general case will be more complicated than shown, as we are plugging together constant-size graph gadgets, and there may be duplication of some subgraphs.)  

Direct-sum composition keeps the maximum degree and norm of the graph under control---each is at most twice its value for the worst single gate.  Therefore the second~\cite{Reichardt09spanprogram} algorithm applies.  However, direct-sum composition also leads to additional overhead.  In particular, a witness in the first graph will be supported only on numbered vertices (note that the graph is bipartite), whereas a witness in the second graph will be supported on some of the internal vertices as well.  This means roughly that the second witness will be harder to detect, because after normalization its overlap on vertex $0$ will be smaller.  Scale both witnesses so that the amplitude on vertex $0$ is one.  The \emph{witness size} ($\wsizeop$) measures the squared length of the witness only on numbered vertices, whereas the \emph{full witness size} ($\wsizefop$) measures the squared length on all vertices.  For~\cite{Reichardt09spanprogram}, it was sufficient to consider only span program witness size, because for canonical span programs like in \figref{f:tensorproductgraphexample} the two measures are equal.  (For technical reasons, we will actually define $\wsizefop$ to be $1 + \wsizeop$ even in this case.)  For our analysis, we will need to bound the full witness size in terms of the witness size.  We maintain this bound in a recursion from the formula's leaves toward its root. 

A span program is called \emph{strict} if every vertex on one half of the bipartite graph is either an input vertex (vertices $1$--$7$ in the graphs of \figref{f:graphexamples}) or the output vertex (vertex $0$).  Thus the first graph in the example above corresponds to a strict span program, and the second does not.  The original definition of span programs, in~\cite{KarchmerWigderson93span}, allowed for only strict span programs.  This was sensible because any other vertices on the input/output part of the graph's bipartition can always be projected away, yielding a strict span program that computes the same function.  For developing time-efficient quantum algorithms, though, it seems important to consider span programs that are not strict.  Unfortunately, going backwards, e.g., from \ref{f:tensorproductgraphexample} to \ref{f:directsumgraphexample}, is probably difficult in general.  

\thmref{t:unbalancedformulaevaluation} does \emph{not} follow from the formula-evaluation techniques of~\cite{ReichardtSpalek08spanprogram}, together with \thmref{t:spanprogramalgorithm} from~\cite{Reichardt09spanprogram}.  This tempting approach falls into intractable technical difficulties.  In particular, the same span program can be used at two vertices $v$ and $w$ in $\varphi$ only if $g_v = g_w$ and the general adversary bounds of $v$'s input subformulas are the same as those for $w$'s inputs up to simultaneous scaling.  In general, then, an almost-balanced formula will require an unbounded number of different span programs.  However, the analysis in~\cite{ReichardtSpalek08spanprogram} loses a factor that depends badly on the individual span programs.  Since the dependence is not continuous, even showing that the span programs in use all lie within a compact set would not be sufficient to obtain an $O(1)$ upper bound.  In contrast, the approach we follow here allows bounding the lost factor by an exponential in $k$, uniformly over different gate imbalances.

\subsection{Quantum algorithm to evaluate approximately balanced AND-OR formulas} \label{s:quantumANDORapproxbalanceresults}

Ambainis et al.~\cite{AmbainisChildsReichardtSpalekZhang07andor} use a weaker balance criterion for AND-OR formulas than \defref{t:approxbalancedef}.  They define an AND-OR formula to be approximately balanced if $\pathinvsum{\varphi} = O(1)$ and $\pathsum{\varphi} = O(n)$.  Here $n$ is the size of the formula, i.e., the number of leaves, and $\pathinvsum{\varphi}$ and $\pathsum{\varphi}$ are defined by: 

\begin{definition} \label{t:approximatelybalanced}
For each vertex $v$ in a formula $\varphi$, let 
\beq\begin{split} \label{e:approximatelybalanced}
\pathinvsum{v} &= \max_\path \, \sum_{w \in \path} \frac1{\ADVpm(\varphi_w)} \\
\pathsum{v} &= \max_\path \, \sum_{w \in \path} \ADVpm(\varphi_w)^2
 \enspace ,
\end{split}\eeq
with each maximum taken over all simple paths $\path$ from $v$ to a leaf.  Let $\sigma_\pm(\varphi) = \sigma_\pm(r)$, where $r$ is the root of $\varphi$.  
\end{definition}

Recall that $\ADVpm(\varphi) = \ADV(\varphi) = \sqrt n$ for an AND-OR formula.  \defref{t:approxbalancedef} is a stricter balance criterion because $\beta$-balance of a formula $\varphi$ implies (by \lemref{t:balancegeometric}) that $\pathinvsum{\varphi}$ and $\pathsum{\varphi}$ are both dominated by geometric series.  However, the same steps followed by the proof of \thmref{t:unbalancedformulaevaluation} still suffice for proving the~\cite{AmbainisChildsReichardtSpalekZhang07andor} result, and, in fact, for strengthening it.  We show: 

\begin{theorem} \label{t:approxbalancedandor}
Let $\varphi$ be an AND-OR formula of size $n$.  Then after polynomial-time classical preprocessing that does not depend on the input $x$, $\varphi(x)$ can be evaluated by a quantum algorithm with error at most $1/3$ using $O\big(\sqrt n \, \pathinvsum{\varphi}\big)$ input queries.  The algorithm's running time is $\sqrt n \, \pathinvsum{\varphi} (\log n)^{O(1)}$ assuming unit-cost coherent access to the preprocessed string.
\end{theorem}

For the special case of AND-OR formulas with $\pathinvsum{\varphi} = O(1)$, \thmref{t:approxbalancedandor} strengthens \thmref{t:unbalancedformulaevaluation}.  The requirement that $\pathinvsum{\varphi} = O(1)$ allows for some gates in the formula to be very unbalanced.  \thmref{t:approxbalancedandor} also strengthens~\cite[Theorem~1]{AmbainisChildsReichardtSpalekZhang07andor} because it does not require that $\pathsum{\varphi} = O(n)$.  For example, a formula that is biased near the root, but balanced at greater depths can have $\pathinvsum{\varphi} = O(1)$ and $\pathsum{\varphi} = \omega(n)$.  By substituting the bound $\pathinvsum{\varphi} = O(\sqrt d)$ for a depth-$d$ formula~\cite[Def.~3]{AmbainisChildsReichardtSpalekZhang07andor}, a corollary of \thmref{t:approxbalancedandor} is that a depth-$d$, size-$n$ AND-OR formula can be evaluated using $O(\sqrt{n d})$ queries.  This improves the depth-dependence from~\cite{AmbainisChildsReichardtSpalekZhang07andor}, and matches the dependence from an earlier version of that article~\cite{ambainis07nand}.  

The essential reason that the \defref{t:approxbalancedef} balance condition can be weakened is that for the specific gates AND and OR, by writing out the optimal span programs explicitly we can prove that they satisfy stronger properties than are necessarily true for other functions.

\section{Span programs}

\subsection{Definitions} \label{s:spanprogramdef}

We briefly recall some definitions from~\cite[Sec.~2]{Reichardt09spanprogram}.  Additionally, we define a span program complexity measure, the full witness size, that charges even for the ``free" inputs.  This quantity is important for developing quantum algorithms that are time efficient as well as query efficient.  

For a natural number $n$, let $[n] = \{1, 2, \ldots, n\}$.  For a finite set $X$, let $\C^X$ be the inner product space $\C^{\abs X}$ with orthonormal basis $\{ \ket x : x \in X \}$.  For vector spaces $V$ and $W$ over $\C$, let $\L(V, W)$ be the set of linear transformations from $V$ into $W$, and let $\L(V) = \L(V, V)$.  For $A \in \L(V, W)$, $\norm{A}$ is the operator norm of $A$.  For a string $x \in \B^n$, let $\bar x$ denote its bitwise complement.  

\begin{definition}[{\cite{HoyerLeeSpalek05compose, HoyerLeeSpalek07negativeadv}}] \label{t:adversarydef}
For finite sets $C$, $E$ and $\D \subseteq C^n$, let $f: \D \rightarrow E$.  An adversary matrix for $f$ is a real, symmetric matrix $\Gamma \in \L(\C^\D)$ that satisfies $\bra x \Gamma \ket y = 0$ whenever $f(x) = f(y)$.  

The general adversary bound for $f$, with costs $s \in [0, \infty)^n$, is 
\begin{equation}
\ADVpm_s(f) = \max_{\Large \substack{\text{adversary matrices $\Gamma$:} \\
\forall j \in [n], \, \norm{\Gamma \circ \Delta_j} \leq s_j}} \norm{\Gamma}
 \enspace .
\end{equation}
Here $\Gamma \circ \Delta_j$ denotes the entry-wise matrix product between $\Gamma$ and $\Delta_j = \sum_{x, y : x_j \neq y_j} \ketbra x y$.  The (nonnegative-weight) adversary bound for $f$, with costs $s$, is defined by the same maximization, except with $\Gamma$ restricted to have nonnegative entries.  In particular, $\ADVpm_s(f) \geq \ADV_s(f)$.  
\end{definition}

Letting $\vec{1} = (1, 1, \ldots, 1)$, the adversary bound for $f$ is $\ADV(f) = \ADV_{\vec 1}(f)$ and the general adversary bound for $f$ is $\ADVpm(f) = \ADVpm_{\vec 1}(f)$.  By~\cite{HoyerLeeSpalek07negativeadv}, $Q(f) = \Omega(\ADVpm(f))$.  

\begin{definition}[Span program~\cite{KarchmerWigderson93span}] \label{t:spanprogramdef}
A span program $P$ consists of a natural number $n$, a finite-dimensional inner product space $V$ over $\C$, a ``target" vector $\ket t \in V$, disjoint sets $\Ifree$ and $I_{j,b}$ for $j \in [n]$, $b \in \B$, and ``input vectors" $\ket{v_i} \in V$ for $i \in \Ifree \cup \bigcup_{j \in [n], b \in \B} I_{j,b}$.  

To $P$ corresponds a function $f_P : \B^n \rightarrow \B$, defined on $x \in \B^n$ by 
\beq \label{e:spanprogramdef}
f_P(x) = \begin{cases}
1 & \text{if $\ket t \in \Span(\{ \ket{ v_i } : i \in \Ifree \cup \bigcup_{j \in [n]} I_{j, x_j} \})$} \\
0 & \text{otherwise}
\end{cases}
\eeq
\end{definition}

Some additional notation is convenient.  Fix a span program $P$.  Let $I = \Ifree \cup \bigcup_{j \in [n], b \in \B} I_{j,b}$.  Let $A \in \L(\C^I, V)$ be given by $A = \sum_{i \in I} \ketbra{v_i}{i}$.  For $x \in \B^n$, let $I(x) = \Ifree \cup \bigcup_{j \in [n]} I_{j, x_j}$ and $\Pi(x) = \sum_{i \in I(x)} \ketbra i i \in \L(\C^I)$.  Then $f_P(x) = 1$ if $\ket t \in \Range(A \Pi(x))$.  A vector $\ket w \in \C^I$ is said to be a witness for $f_P(x) = 1$ if $\Pi(x) \ket w = \ket w$ and $A \ket w = \ket t$.  A vector $\ket{w'} \in V$ is said to be a witness for $f_P(x) = 0$ if $\braket t {w'} = 1$ and $\Pi(x) A^\dagger \ket{w'} = 0$.  

\begin{definition}[Witness size] \label{t:wsizedef}
Consider a span program $P$, and a vector $s \in [0, \infty)^n$ of nonnegative ``costs."  Let $S = \sum_{j \in [n], b \in \B, i \in I_{j,b}} \sqrt{s_j} \ketbra i i \in \L(\C^I)$.  For each input $x \in \B^n$, define the witness size of $P$ on $x$ with costs $s$, $\wsizexS P x s$, as follows: 
\beq \label{e:wsizedef}
\wsizexS P x s = \begin{cases}
\min_{\ket w : \, A \Pi(x) \ket w = \ket t} \norm{S \ket w}^2 & \text{if $f_P(x) = 1$} \\
\min_{\substack{\ket{w'} : \, \braket{t}{w'} = 1 \\ \Pi(x) A^\adjoint \ket{w'} = 0}} \norm{S A^\adjoint \ket{w'}}{}^2 & \text{if $f_P(x) = 0$}
\end{cases}
\eeq

The witness size of $P$ with costs $s$ is 
\beq
\wsizeS P s =  \max_{x \in \B^n} \wsizexS P x s
 \enspace .
\eeq

Define the full witness size $\wsizefS P s$ by letting $\Sf = S + \sum_{i \in \Ifree} \ketbra i i$ and 
\begin{align}
\wsizefxS P x s &= \begin{cases}
\min_{\ket w : \, A \Pi(x) \ket w = \ket t} (1 + \norm{\Sf \ket w}{}^2) & \text{if $f_P(x) = 1$} \\
\min_{\substack{\ket{w'} : \, \braket{t}{w'} = 1 \\ \Pi(x) A^\adjoint \ket{w'} = 0}} (\norm{\ket{w'}}{}^2 + \norm{S A^\adjoint \ket{w'}}{}^2) & \text{if $f_P(x) = 0$}
\end{cases} \\
\wsizefS P s &=  \max_{x \in \B^n} \wsizefxS P x s
 \enspace .
\end{align}
\end{definition}

When the subscript $s$ is omitted, the costs are taken to be uniform, $s = \vec 1 = (1, 1, \ldots, 1)$, e.g., $\wsizef P = \wsizefS P {\vec 1}$.  The witness size is defined in~\cite{ReichardtSpalek08spanprogram}.  The full witness size is defined in~\cite[Sec.~8]{Reichardt09spanprogram}, but is not named there.  A \emph{strict} span program has $\Ifree = \emptyset$, so $\Sf = S$, and a \emph{monotone} span program has $I_{j, 0} = \emptyset$ for all $j$~\cite[Def.~4.9]{Reichardt09spanprogram}.

\subsection{Quantum algorithm to evaluate a span program based on its full witness~size}

\cite[Theorem~9.3]{Reichardt09spanprogram} gives a quantum query algorithm for evaluating span programs based on the full witness size.  The algorithm is based on a quantum walk on a certain graph.  Provided that the degree of the graph is not too large, it can actually be implemented efficiently.  

\begin{theorem}[{\cite[Theorem~9.3]{Reichardt09spanprogram}}] \label{t:generalspanprogramalgorithmnonblackbox}
Let $P$ be a span program.  Then $f_P$ can be evaluated using 
\beq \label{e:generalspanprogramalgorithmnonblackbox}
T = O\big( \wsizef P \, \norm{\abst(A_{G_P})} \big)
\eeq
quantum queries, with error probability at most $1/3$.  Moreover, if the maximum degree of a vertex in $G_P$ is $d$, then the time complexity of the algorithm for evaluating $f_P$ is at most a factor of $(\log d) \big(\log (T \log d) \big)^{O(1)}$ worse, after classical preprocessing and assuming constant-time coherent access to the preprocessed string.  
\end{theorem}

\begin{proof}[Proof sketch]
The query complexity claim is actually slightly weaker than~\cite[Theorem~9.3]{Reichardt09spanprogram}, which allows the target vector to be scaled downward by a factor of $\sqrt{\wsizef P}$.  

The time-complexity claim will follow from the proof of~\cite[Theorem~9.3]{Reichardt09spanprogram}, in~\cite[Prop.~9.4, Theorem~9.5]{Reichardt09spanprogram}.  The algorithm for evaluating $f_P(x)$ uses a discrete-time quantum walk on the graph $G_P(x)$.  If the maximum degree of a vertex in $G_P$ is $d$, then each coin reflection can be implemented using $O(\log d)$ single-qubit unitaries and queries to the preprocessed string~\cite{GroverRudolph02superposition, ChiangNagajWocjan09simulate}.  Finally, the $\big(\log (T \log d) \big)^{O(1)}$ factor comes from applying the Solovay-Kitaev Theorem~\cite{ksw:qc-book} to compile the single-qubit unitaries into products of elementary gates, to precision $1/O(T \log d)$.  
\end{proof}

We remark that together with~\cite[Theorem~3.1]{Reichardt09spanprogram}, \thmref{t:generalspanprogramalgorithmnonblackbox} gives a way of transforming a one-sided-error quantum 
algorithm into a span program, and back into a quantum algorithm, such that the time complexity is nearly preserved, after preprocessing.  This is only a weak equivalence, because aside from requiring preprocessing the algorithm from \thmref{t:generalspanprogramalgorithmnonblackbox} also has two-sided error.  To some degree, though, it complements the equivalence results for best span program witness size and bounded-error quantum \emph{query} complexity~\cite[Theorem~7.1, Theorem~9.2]{Reichardt09spanprogram}.

\subsection{Direct-sum span program composition}

Let us study the full witness size of the direct-sum composition of span programs.  We begin by recalling the definition of direct-sum composition.  

Let $f : \B^n \rightarrow \B$ and $S \subseteq [n]$.  For $j \in [n]$, let $m_j$ be a natural number, with $m_j = 1$ for $j \notin S$.  For $j \in S$, let $f_j : \B^{m_j} \rightarrow \B$.  Define $y : \B^{m_1} \times \cdots \times \B^{m_n} \rightarrow \B^n$ by 
\beq
y(x)_j = \begin{cases} f_j(x_j) & \text{if $j \in S$} \\ x_j & \text{if $j \notin S$} \end{cases}
\eeq  
Define $g : \B^{m_1} \times \cdots \times \B^{m_n} \rightarrow \B$ by $g(x) = f(y(x))$.  For example, if $S = [n] \smallsetminus \{1\}$, then 
\beq
g(x) = f\big(x_1, f_2(x_2), \ldots, f_n(x_n)\big)
 \enspace .
\eeq
Given span programs for the individual functions $f$ and $f_j$ for $j \in S$, we will construct a span program for $g$.  We remark that although we are here requiring that the inner functions $f_j$ act on disjoint sets of bits, this assumption is not necessary for the definition.  It simplifies the notation, though, for the cases $S \neq [n]$, and will suffice for our applications.  

Let $P$ be a span program computing $f_P = f$.  Let $P$ have inner product space $V$, target vector $\ket t$ and input vectors $\ket{v_i}$ indexed by $\Ifree$ and $I_{jc}$ for $j \in [n]$ and $c \in \B$.  

For $j \in [n]$, let $s_j \in [0, \infty)^{m_j}$ be a vector of costs, and let $s \in [0, \infty)^{\sum m_j}$ be the concatenation of the vectors $s_j$.  For $j \in S$, let $P^{j0}$ and $P^{j1}$ be span programs computing $f_{P^{j1}} = f_j : \B^{m_j} \rightarrow \B$ and $f_{P^{j0}} = \neg f_j$, with $r_j = \wsizeS{P^{j0}}{s_j} = \wsizeS{P^{j1}}{s_j}$.  For $c \in \B$, let $P^{jc}$ have inner product space $V^{jc}$ with target vector $\ket{t^{jc}}$ and input vectors indexed by $\Ifree^{jc}$ and $I^{jc}_{kb}$ for $k \in [m_j]$, $b \in \B$.  For $j \notin S$, let $r_j = s_j$.  

Let $I_S = \bigcup_{j \in S, c \in \B} I_{jc}$.  Define $\jc : I_S \rightarrow [n] \times \B$ by $\jc(i) = (j,c)$ if $i \in I_{jc}$.  The idea is that $\jc$ maps $i$ to the input span program that must evaluate to $1$ in order for $\ket{v_i}$ to be available in~$P$.  

There are several ways of composing the span programs $P$ and $P^{jc}$ to obtain a span program $Q$ computing the composed function $f_Q = g$ with $\wsizeS Q s \leq \wsizeS P r$~\cite[Defs.~4.4, 4.5, 4.6]{Reichardt09spanprogram}.  We focus on direct-sum composition.  

\begin{definition}[{\cite[Def.~4.5]{Reichardt09spanprogram}}] \label{t:directsumcomposedef}
The direct-sum-composed span program $Q^\oplus$ is defined by: 
\begin{itemize}
\item
The inner product space is $V^\oplus = V \oplus \bigoplus_{j \in S, c \in \B} (\C^{I_{jc}} \otimes V^{jc})$.  Any vector in $V^\oplus$ can be uniquely expressed as $\ket{u}_V + \sum_{i \in I_S} \ket{i} \otimes \ket{u_i}$, where $\ket u \in V$ and $\ket{u_i} \in V^{\jc(i)}$.  
\item
The target vector is $\ket{t^\oplus} = \ket{t}_V$.  
\item
The free input vectors are indexed by $\Ifree^\oplus = \Ifree \cup I_S \cup \bigcup_{j \in S, c \in \B} (I_{jc} \times \Ifree^{jc})$ with, for $i \in \Ifree^\oplus$, 
\beq
\ket{v^\oplus_i} = \begin{cases}
\ket{v_i}_V & \text{if $i \in \Ifree$} \\
\ket{v_i}_V - \ket i \otimes \ket{t^{jc}} & \text{if $i \in I_{jc}$ and $j \in S$} \\
\ket{i'} \otimes \ket{v_{i''}} & \text{if $i = (i', i'') \in I_{jc} \times \Ifree^{jc}$}
\end{cases}
\eeq
\item
The other input vectors are indexed by $I^\oplus_{(jk)b}$ for $j \in [n]$, $k \in [m_j]$, $b \in \B$.  For $j \notin S$, $I^\oplus_{(j1)b} = I_{jb}$, with $\ket{v^\oplus_i} = \ket{v_i}_V$ for $i \in I^\oplus_{(j1)b}$.  For $j \in S$, let $I^\oplus_{(jk)b} = \bigcup_{c \in \B} (I_{jc} \times I^{jc}_{kb})$.  For $i \in I_{jc}$ and $i' \in I^{jc}_{kb}$, let 
\beq
\ket{v^\oplus_{ii'}} = \ket{i} \otimes \ket{v_{i'}}
 \enspace .
\eeq
\end{itemize}
\end{definition}

By~\cite[Theorem~4.3]{Reichardt09spanprogram}, $f_{Q^\oplus} = g$ and $\wsizeS {Q^\oplus} s \leq \wsizeS P r$.  (While that theorem is stated only for the case $S = [n]$, it is trivially extended to other $S \subset [n]$.)  We give a bound on how quickly the full witness size can grow relative to the witness size: 

\begin{lemma} \label{t:wsizefcompose}
Under the above conditions, for each input $x \in \B^{m_1} \times \cdots \times \B^{m_n}$, with $y = y(x)$, 
\begin{itemize}
\item 
If $g(x) = 1$, let $\ket w$ be a witness to $f_P(y) = 1$ such that $\sum_{j \in [n], i \in I_{j y_j}} r_j \abs{w_i}^2 = \wsizexS P y r$.  Then 
\beq\begin{split} \label{e:wsizefcomposetrue}
\frac{ \wsizefxS {Q^\oplus} x s }{ \wsizexS P y r }
&\leq 
\sigma\big(y, \ket w\big) + \frac{1 + \sum_{i \in \Ifree} \abs{w_i}^2}{\wsizexS P y r} \\
&\text{where $\sigma(y, \ket w) = \max_{\substack{j \in S : \\ \text{$\exists i \in I_{j y_j}$ with $\braket{i}{w} \neq 0$}}} \frac{\wsizefS {P^{j y_j}} {s_j}}{\wsizeS {P^{j y_j}} {s_j}}$}
 \enspace .
\end{split}\eeq
\item 
If $g(x) = 0$, let $\ket{w'}$ be a witness to $f_P(y) = 0$ such that $\sum_{j \in [n], i \in I_{j \bar y_j}} r_j \abs{\braket{w'}{v_i}}^2 = \wsizexS P y r$.  Then 
\beq\begin{split} \label{e:wsizefcomposefalse}
\frac{ \wsizefxS {Q^\oplus} x s }{ \wsizexS P y r }
&\leq 
\sigma(\bar y, \ket{w'}) + \frac{\norm{\ket{w'}}^2}{\wsizexS P y r} \\
&\text{where $\sigma(\bar y, \ket{w'}) = \max_{\substack{j \in S : \\ \text{$\exists i \in I_{j \bar y_j}$ with $\braket{v_i}{w'} \neq 0$}}} \frac{\wsizefS {P^{j \bar y_j}} {s_j}}{\wsizeS {P^{j \bar y_j}} {s_j}}$}
 \enspace .
\end{split}\eeq
\end{itemize}
If $S = \emptyset$, then $\sigma(y, \ket w)$ and $\sigma(\bar y, \ket{w'})$ should each be taken to be $1$ in the above equations.  
\end{lemma}

\begin{proof}
We follow the proof of~\cite[Theorem~4.3]{Reichardt09spanprogram}, except keeping track of the full witness size.  
Note that if $S = \emptyset$, then Eqs.~\eqnref{e:wsizefcomposetrue} and~\eqnref{e:wsizefcomposefalse} are immediate by definition of $\wsizefxS {Q^\oplus} x s$.  

Let $I(y)' = I(y) \smallsetminus \Ifree = \bigcup_{j \in [n]} I_{j y_j}$.  

In the first case, $g(x) = 1$, for $j \in S$ let $\ket{w^{j y_j}} \in \C^{I^{j y_j}}$ be a witness to $f_{P^{j y_j}}(x_j) = 1$ such that $\wsizefxS{P^{j y_j}}{x_j}{s} = 1 + \sum_{i \in \Ifree^{jy_j}} \abs{w^{j y_j}_i}^2 + \sum_{k \in [m_j], i \in I^{j y_j}_{k (x_j)_k}} (s_j)_k \abs{w^{j y_j}_i}^2$.  As in~\cite[Theorem~4.3]{Reichardt09spanprogram}, let $\ket{w^\oplus} \in \C^{I^\oplus(x)}$ be given by 
\beq
w^\oplus_i 
= \begin{cases}
w_i & \text{if $i \in I(y)$} \\
w_{i'} w^{\jc(i')}_{i''} & \text{if $i = (i', i'')$ with $i' \in I(y)' \cap I_S$, $i'' \in I^{\jc(i')}(x)$} \\
0 & \text{otherwise}
\end{cases}
\eeq
Then $\ket{w^\oplus}$ is a witness for $f_{Q^\oplus}(x) = 1$, and we compute 
\begin{align}
\wsizefxS{Q^\oplus}{x}{s}
&\leq 1 + \sum_{i \in \Ifree^\oplus} \abs{w^\oplus_i}^2 + \sum_{\substack{j \in [n], k \in [m_j], \\ i \in I^\oplus_{(jk)(x_j)_k}}} (s_j)_k \abs{w^\oplus_i}^2 \nonumber \\
&= 1 + \sum_{i \in \Ifree} \abs{w_i}^2 + \sum_{j \in [n] \smallsetminus S, i \in I_{j x_j}} s_j \abs{w_i}^2 \\ &\qquad + \sum_{j \in S, i \in I_{j y_j}} \abs{w_i}^2 \Bigg( 1 + \sum_{i' \in \Ifree^{j y_j}} \abs{w^{j y_j}_{i'}}{}^2 + \sum_{k \in [m_j], i' \in I^{j y_j}_{k (x_j)_k}} (s_j)_k \abs{w^{j y_j}_{i'}}{}^2 \Bigg) \nonumber \\
&= 1 + \sum_{i \in \Ifree} \abs{w_i}^2 + \sum_{j \in [n] \smallsetminus S, i \in I_{j x_j}} s_j \abs{w_i}^2 + \sum_{j \in S, i \in I_{j y_j}} \abs{w_i}^2 \, \wsizefxS{P^{j y_j}}{x_j}{s_j} \nonumber 
 \enspace .
\end{align}
Eq.~\eqnref{e:wsizefcomposetrue} follows using the bound $\wsizefxS{P^{jy_j}}{x_j}{s_j} \leq \sigma(y, \ket w) r_j$ for $j \in S$, and $s_j = r_j$ for $j \notin S$.  

Next consider the case $g(x) = 0$.  For $j \in S$, let $\ket{u^{j \bar y_j}} \in V^{j \bar y_j}$ be a witness for $f_{P^{j \bar y_j}}(x_j) = 0$ with $\wsizefxS{P^{j \bar y_j}} {x_j} s = \norm{\ket{u^{j \bar y_j}}}{}^2 + \sum_{k \in [m_j], i \in I^{j \bar y_j}_{k \overline{(x_j)_k}}} (s_j)_k \abs{\braket{v_i}{u^{j \bar y_j}}}{}^2$.  As in~\cite[Theorem~4.3]{Reichardt09spanprogram},~let 
\beq
\ket{u^\oplus} = \ket{w'}_V + \sum_{i \in I_S \smallsetminus I(y)} \braket{v_i}{w'} \ket i \otimes \ket{u^{\jc(i)}}
 \enspace .
\eeq
Then $\ket{u^\oplus}$ is a witness for $f_{Q^\oplus}(x) = 0$, and, moreover, 
\begin{align}
\wsizefxS{Q^\oplus}{x}{s}
&\leq \norm{\ket{u^\oplus}}{}^2 + \sum_{j \in [n], k \in [m_j], i \in I^\oplus_{(jk)\overline{(x_j)_k}}} (s_j)_k \abs{\braket{v^\oplus_i}{u^\oplus}}{}^2 \nonumber \\
&= \norm{\ket{u^\oplus}}{}^2 + \sum_{\substack{j \in [n] \smallsetminus S \\ i \in I_{j \bar x_j}}} s_j \abs{\braket{v^\oplus_i}{u^\oplus}}{}^2 + \sum_{\substack{j \in S, k \in [m_j], \\ i \in I_{j \bar y_j}, i' \in I^{j \bar y_j}_{k \overline{(x_j)_k}}}} (s_j)_k \abs{\braket{v^\oplus_{ii'}}{u^\oplus}}{}^2 \nonumber \\
&= \norm{\ket{w'}}^2 + \sum_{\substack{j \in [n] \smallsetminus S \\ i \in I_{j \bar x_j}}} s_j \abs{\braket{v_i}{w'}}{}^2 \\ &\qquad + \sum_{j \in S, i \in I_{j \bar y_j}} \abs{\braket{v_i}{w'}}^2 \Bigg( \norm{\ket{u^{j \bar y_j}}}{}^2 + \sum_{k \in [m_j], i' \in I^{j \bar y_j}_{k \overline{(x_j)_k}}} (s_j)_k \abs{\braket{v_{i'}}{u^{j \bar y_j}}}{}^2 \Bigg) \nonumber \\
&= \norm{\ket{w'}}^2 + \sum_{\substack{j \in [n] \smallsetminus S \\ i \in I_{j \bar x_j}}} r_j \abs{\braket{v_i}{w'}}{}^2 + \sum_{j \in S, i \in I_{j \bar y_j}} \abs{\braket{v_i}{w'}}^2 \, \wsizefxS{P^{j \bar y_j}}{x_j}{s_j} \nonumber
 \enspace .
\end{align}
Eq.~\eqnref{e:wsizefcomposefalse} follows using the bound $\wsizefxS{P^{j \bar y_j}}{x_j}{s_j} \leq \sigma(\bar y, \ket{w'}) r_j$ for $j \in S$.  
\end{proof}

\lemref{t:wsizefcompose} is a key step in the formula-evaluation results in this article and~\cite{Reichardt09andorfaster}.  It is used to track the full witness size for span programs recursively composed in a direct-sum manner along a formula.  The proof of \thmref{t:unbalancedformulaevaluation} will require the lemma with the weaker bounds $\sigma(y, \ket w), \sigma(\bar y, \ket{w'}) \leq \max_{j \in S, c \in \B} \wsizefS {P^{jc}} {s_j} / \wsizeS {P^{jc}} {s_j}$.  \thmref{t:approxbalancedandor} will use only the slightly stronger bounds $\sigma(y, \ket w) \leq \max_{j \in S} \wsizefS {P^{j y_j}} {s_j} / \wsizeS {P^{j y_j}} {s_j}$, $\sigma(\bar y, \ket{w'}) \leq \max_{j \in S} \wsizefS {P^{j \bar y_j}} {s_j} / \wsizeS {P^{j \bar y_j}} {s_j}$.  However, the proof of~\cite[Theorem~1.1]{Reichardt09andorfaster} will require the bounds of Eqs.~\eqnref{e:wsizefcomposetrue} and~\eqnref{e:wsizefcomposefalse}.

\section{Evaluation of almost-balanced formulas} \label{s:approxbalancedproof}

In this section, we will apply the span program framework from~\cite{Reichardt09spanprogram} to prove \thmref{t:unbalancedformulaevaluation}.  Our algorithm will be given by applying \thmref{t:generalspanprogramalgorithmnonblackbox} to a certain span program.  Before beginning the proof, though, we will give two necessary lemmas.  

Consider a span program $P$ with corresponding weighted graph $G_P$, from~\cite[Def.~8.2]{Reichardt09spanprogram}.  We will need a bound on the operator norm of $\abst(A_{G_{P_v}})$, the entry-wise absolute value of the weighted adjacency matrix $A_{G_{P_v}}$.  If $P$ is \emph{canonical}~\cite[Def.~5.1]{Reichardt09spanprogram}, then we can indeed obtain such a bound in terms of the witness size of $P$: 

\begin{lemma} \label{t:abstAGPnorm}
Let $s \in (0, \infty)^k$, and let $P$ be a canonical span program computing a function $f : \{0,1\}^k \rightarrow \{0,1\}$ with input vectors indexed by the set $I$.  
Assume that for each $x \in \{0,1\}^k$ with $f(x) = 0$, an optimal witness to $f_P(x) = 0$ is $\ket x$ itself.  
Then 
\beq \label{e:abstAGPnorm}
\norm{\abst(A_{G_P})} \leq 2^k \Big( 1 + \frac{\wsizeS P s}{\min_{j \in [k]} s_j} \Big) + \abs I
 \enspace .
\eeq
\end{lemma}

\begin{proof}
Recall from~\cite[Def.~5.1]{Reichardt09spanprogram}, that $P$ being in canonical form implies that its target vector is $\ket t = \sum_{x : f(x) = 0} \ket x$, and that the matrix $A$ whose columns are the input vectors of $P$ can be expressed as 
\beq
A
= \sum_{i \in I} \ketbra{v_i}{i}
= \sum_{j \in [k], \, x : f(x) = 0} \ketbra{x}{j, \bar x_j} \otimes \bra{v_{xj}}
 \enspace .
\eeq
By assumption, for each $x \in f^{-1}(0)$, 
\beq
\sum_{j \in [k]} s_j \norm{\ket{v_{xj}}}^2 = \wsizexS P x s \leq \wsizeS P s
 \enspace .
\eeq
In particular, letting $\sigma = \min_{j \in [k]} s_j > 0$, we can bound 
\beq\begin{split} \label{e:abstAGPnormstep}
\sum_{j \in [k]} \norm{\ket{v_{xj}}}^2
&\leq
\frac 1 \sigma \sum_{j \in [k]} s_j \norm{\ket{v_{xj}}}^2 \\
&\leq
\frac{\wsizeS P s}{\sigma}
 \enspace .
\end{split}\eeq

The rest of the argument follows from the definition of the weighted adjacency matrix $A_{G_P}$.  From~\cite[Def.~8.1, Prop.~8.8]{Reichardt09spanprogram}, $\norm{\abst(A_{G_P})} \leq \norm{\abst(\biadj_{G_P})}^2$, where $\biadj_{G_P}$ is the \emph{biadjacency} matrix corresponding to $P$, 
\beq
\biadj_{G_P} = \left( \begin{matrix} \ket t & A \\ 0 & \identity \end{matrix} \right)
 \enspace ,
\eeq
and $\identity$ is an $\abs I \times \abs I$ identity matrix.  
Now bound $\norm{\abst(\biadj_{G_P})}$ by its Frobenius norm: 
\beq\begin{split}
\norm{\abst(A_{G_P})} 
&\leq \norm{\abst(\biadj_{G_P})}^2 \\
&\leq \norm{\abst(\biadj_{G_P})}_F^2 \\
&= \norm{\ket t}^2 + \sum_{\substack{x : f(x) = 0, \\ j \in [k]}} \norm{\ket{v_{xj}}}^2 + \abs I \\
&\leq 2^k + 2^k \max_{x : f(x) = 0} \sum_{j \in [k]} \norm{\ket{v_{xj}}}^2 + \abs I
 \enspace .
\end{split}\eeq
Eq.~\eqnref{e:abstAGPnorm} follows by substituting in Eq.~\eqnref{e:abstAGPnormstep}.  
\end{proof}

An important quantity in the proof of \thmref{t:unbalancedformulaevaluation} will be $\pathinvsum{\varphi}$, from \defref{t:approximatelybalanced}.  For an almost-balanced formula $\varphi$, $\pathinvsum{\varphi} = O(1)$.  

\begin{lemma} \label{t:balancegeometric}
Consider a $\beta$-balanced formula $\varphi$ over a gate set $\cS$ in which every gate depends on at least two input bits.  Then for every vertex $v$, with children $c_1, c_2, \ldots, c_k$, 
\beq \label{e:balancegeometric}
\frac{\ADVpm(\varphi_v)}{\max_j \ADVpm(\varphi_{c_j})} \geq \sqrt{1+\frac{1}{\beta^2}}
 \enspace .
\eeq
In particular, 
\beq \label{e:balancegeometricpathinvsum}
\pathinvsum{\varphi} \leq (2+\sqrt 2)\beta^2
 \enspace .
\eeq
\end{lemma}

\begin{proof}
Consider a vertex $v$ with corresponding gate $g = g_v : \{0,1\}^k \rightarrow \{0,1\}$.  By \thmref{t:adversarycomposition}, $\ADVpm(\varphi_v) = \ADVpm_s(g)$, where $s_j = \ADVpm(\varphi_{c_j})$.  It is immediate from the definitions that $\ADVpm_s(g) \geq \ADV_s(g)$.  We will show that $\ADV_s(g) \geq \sqrt{1+1/\beta^2} (\max_j s_j)$, using that $\max_j s_j / \min_j s_j \leq \beta$.  

Use the weighted minimax formulation of the adversary bound from~\cite[Theorem~18]{HoyerLeeSpalek07negativeadv}: 
\beq \label{e:minimaxadv}
\ADV_s(g) = \min_{p} \max_{\substack{x, y \in \{0,1\}^k\\ g(x) \neq g(y)}} \frac{1}{\sum_{j : x_j \neq y_j} \sqrt{p_x(j) p_y(j)} / s_j}
 \enspace ,
\eeq
where the minimization is over all choices of probability distributions $p_x$ over $[k]$ for $x \in \{0,1\}^k$.  

Since the adversary bound is monotone increasing in each weight, the worst case is when all but one of the weights are equal to $\max_j s_j / \beta$.  Since for a scalar $c$, $\ADV_{c s}(g) = c \ADV_s(g)$, we may scale so that one weight is $\beta$ and all other weights are $1$.  Assume that the first weight is $s_1 = \beta$; the other $k-1$ cases, $s_2 = \beta$ and so on, are symmetrical.  Assume also that $g$ depends on the first bit; otherwise $\ADVpm_s(g)$ will not depend on $s_1$ so one of the other cases will be worse.  Therefore, there exist inputs $x, y \in \{0,1\}^k$ that differ only on the first bit, but for which $g(x) \neq g(y)$.  

Since the function $g$ depends on at least two input bits, there also exists a third input $z \in \{0,1\}^k$ with $x_1 = z_1$ but $g(z) = g(y) \neq g(x)$.  Indeed, if $g(z) = g(x)$ for every $z$ with $z_1 = x_1$, and if $g(z) = g(y)$ for every $z$ with $z_1 = y_1$, then $g$ depends only on the first bit.  

By Eq.~\eqnref{e:minimaxadv}, 
\beq
\ADVpm_s(g) 
\geq
\min_{p_x, p_y, p_z} \max \Big\{ \frac{1}{\sqrt{p_x(1) p_y(1)} / s_1}, \frac{1}{\sum_{\substack{j \geq 2 \\ x_j \neq z_j}} \sqrt{p_x(j) p_z(j)} / s_j} \Big\}
\eeq
where the minimization is over only the three probability distributions $p_x$, $p_y$ and $p_z$.  In the above expression, we may clearly take $p_y(1) = 1$ and $p_y(j) = 0$ for $j \geq 2$.  We may also use the Cauchy-Schwarz inequality to bound the second term above, and finally substitute $s_1 = \beta$, $s_j = 1$ for $j \geq 2$ to obtain,
\beq
\ADVpm_s(g)
\geq
\min_{p_x} \max \Big\{ \frac{\beta}{\sqrt{p_x(1)}}, \frac{1}{\sqrt{\sum_{j \geq 2} p_x(j)}} \Big\}
 \enspace .
\eeq
The optimum is achieved for $p_x(1) = \beta^2 / (1 + \beta^2)$, so $\ADVpm_s(g) \geq \sqrt{1 + \beta^2}$, as claimed.  

To derive Eq.~\eqnref{e:balancegeometricpathinvsum}, note that $\beta \geq 1$ necessarily.  Then the sum $\pathinvsum{\varphi}$ is dominated by the geometric series 
\beq
\sum_{k=0}^{\infty} \Big( 1 + \frac 1 {\beta^2} \Big)^{-k/2}
 \enspace ,
\eeq
which is at most $(2+\sqrt 2)\beta^2$, with equality at $\beta = 1$.  
\end{proof}

Note that the $1$-balanced formulas over $\cS = \{ \OR_2 \}$ satisfy the inequality~\eqnref{e:balancegeometric} with equality and come arbitrarily close to saturating the inequality~\eqnref{e:balancegeometricpathinvsum}.  

With \lemref{t:abstAGPnorm} and \lemref{t:balancegeometric} in hand, we are ready to prove \thmref{t:unbalancedformulaevaluation}.  

\begin{proof}[Proof of \thmref{t:unbalancedformulaevaluation}]
First of all, we may assume without loss of generality that every gate in $\cS$ depends on at least two input bits.  Indeed, if a gate $g : \{0,1\}^k \rightarrow \{0,1\}$ depends on no input bits, i.e., is the constant $0$ or constant $1$ function, then $g$ can be eliminated from any formula over $\cS$ without changing the adversary balance condition, since $\ADVpm_s(g) = 0$ for all cost vectors $s \in [0,\infty)^k$.  If a gate $g : \{0,1\}^k \rightarrow \{0,1\}$ depends only on one input bit, say the first bit, then $\ADVpm_s(g) = s_1$ for all cost vectors $s$, and therefore similarly $g$ can be eliminated without affecting the adversary balance condition.  

Consider $\varphi$ an $n$-variable, $\beta$-balanced, read-once formula over the finite gate set~$\cS$.  Let $r$ be the root of $\varphi$.  We begin by recursively constructing a span program $P_\varphi$ that computes $\varphi$ and has witness size $\wsize{P_\varphi} = \ADVpm(\varphi)$.  $P_\varphi$ is constructed using direct-sum composition of span programs for each node in $\varphi$.  (Direct-sum composition is also the composition method used in~\cite{ReichardtSpalek08spanprogram}.)  

The construction works recursively, starting at the leaves of $\varphi$ and moving toward the root.  Consider an internal vertex $v$, with children $c_1, \ldots, c_k$.  Let $\alpha_j = \ADVpm(\varphi_{c_j})$, where $\varphi_{c_j}$ is the subformula of $\varphi$ rooted at $c_j$ (\defref{t:approxbalancedef}).  In particular, if $c_j$ is a leaf, then $\alpha_j = 1$.  Assume that for $j \in [k]$ we have inductively constructed span programs $P_{\varphi_{c_j}}$ and $P_{\varphi_{c_j}}^\dagger$ computing $\varphi_{c_j}$ and $\neg \varphi_{c_j}$, respectively, with $\wsize{P_{\varphi_{c_j}}} = \wsize{P_{\varphi_{c_j}}^\dagger} = \alpha_j$.  Apply \cite[Theorem~6.1]{Reichardt09spanprogram}, a generalization of \thmref{t:spanprogramSDPintro}, twice to obtain span programs $P_v$ and $P_v^\dagger$ computing $f_{P_v} = g_v$ and $f_{P_v^\dagger} = \neg g_v$, with $\wsizeS {P_v} \alpha = \wsizeS {P_v^\dagger} \alpha = \ADVpm_\alpha(g_v) = \ADVpm(\varphi_v)$.  

Then let $P_{\varphi_v}$ and $P_{\varphi_v}^\dagger$ be the direct-sum-composed span programs of $P_v$ and $P_v^\dagger$, respectively, with the span programs $P_{\varphi_{c_j}}$, $P_{\varphi_{c_j}}^\dagger$ according to the formula $\varphi$.  By definition of direct-sum composition, the graph $G_{P_{\varphi_v}}$ is built by replacing the input edges of $G_{P_v}$ with the graphs $G_{P_{\varphi_{c_j}}}$ or $G_{P_{\varphi_{c_j}}^\dagger}$; and similarly for $G_{P_{\varphi_v}^\dagger}$.  Some examples are given in~\cite[App.~B]{Reichardt09spanprogram} and in~\cite{ReichardtSpalek08spanprogram}.  By~\cite[Theorem~4.3]{Reichardt09spanprogram}, $P_{\varphi_v}$ (resp.~$P_{\varphi_v}^\dagger$) computes $\varphi_v$ ($\neg \varphi_v$) with $\wsize{P_{\varphi_v}} = \wsize{P_{\varphi_v}^\dagger} = \ADVpm(\varphi_v)$.  

Let $P_\varphi = P_{\varphi_r}$.  We wish to apply \thmref{t:generalspanprogramalgorithmnonblackbox} to $P_\varphi$ to obtain a quantum algorithm, but to do so will need some more properties of the span programs $P_v$ and $P_v^\dagger$.  Recall from~\cite[Theorem~5.2]{Reichardt09spanprogram} that each $P_v$ may be assumed to be in canonical form, satisfying in particular that for any input $y \in \B^k$ with $g_v(y) = 0$ an optimal witness is $\ket y \in \C^{g_v^{-1}(0)}$ itself.  Therefore, \lemref{t:abstAGPnorm} applies, and we obtain 
\beq
\norm{\abst(A_{G_{P_v}})} = 2^k \bigg( 1 + \frac{\wsizeS {P_v} \alpha}{\min_j \alpha_j} \bigg) + \abs I
 \enspace ,
\eeq
where $\abs I$ is the number of input vectors in $P_v$.  
Now use 
\beq\begin{split}
\frac{\wsizeS {P_v} \alpha}{\min_j \alpha_j}
&=
\frac{\max_j \alpha_j}{\min_j \alpha_j} \frac{\ADVpm_\alpha(g_v)}{\max_j \alpha_j} \\
&\leq
\beta k
 \enspace ,
\end{split}\eeq
where we have applied Eq.~\eqnref{e:approxbalancedef} and also $\ADVpm_\alpha(g_v) / \max_j \alpha_j \leq \ADVpm(g_v) \leq k$.  Additionally, by~\cite[Lemma~6.6]{Reichardt09spanprogram}, we may assume that $\abs I \leq 2 k^2 2^k$.  Thus 
\beq
\norm{\abst(A_{G_{P_v}})} = \beta \, 2^{O(k)}
 \enspace .
\eeq
By repeating this argument for the negated function $\neg g_v$ computed by a dual span program $P_v^\dagger$ (\cite[Lemma~4.1]{Reichardt09spanprogram}), we also have $\norm{\abst(A_{G_{P_v^\dagger}})} = \beta \, 2^{O(k)}$.  

A consequence is that 
\beq \label{e:unbalancedformulaevaluationnorm}
\norm{\abst(A_{G_{P_\varphi}})} = \beta \, 2^{O(k_{\text{max}})}
\eeq
where $k_{\text{max}}$ is the maximum fan-in of any gate used in $\varphi$.  Indeed, $G_{P_\varphi}$ is built by ``plugging together" the graphs $G_{P_v}$ and $G_{P_v^\dagger}$ for the different vertices~$v$.  Split the graph $G_{P_\varphi}$ into two pieces, $G_0$ and $G_1$, comprising those subgraphs $G_{P_v}$ and $G_{P_v^\dagger}$ for which the distance of $v$ from $r$ is even or odd, respectively.  Then $\norm{\abst(A_{G_{P_\varphi}})} \leq \norm{\abst(A_{G_0})} + \norm{\abst(A_{G_1})}$.  Since each $G_b$ is the disconnected union of graphs $G_{P_v}$ and $G_{P_v^\dagger}$, $\norm{\abst(A_{G_b})} \leq \max_v \max \{ \norm{\abst(A_{G_{P_v}})}, \norm{\abst(A_{G_{P_v^\dagger}})} \}$.  

Let us bound the full witness size of $P_\varphi$.  

\begin{lemma} \label{t:unbalancedformulaevaluationtruefalsewitness}
Let $v$ be a vertex of $\varphi$.  Then 
\beq
\max \big\{ \wsizef {P_{\varphi_v}}, \wsizef {P_{\varphi_v}^\dagger} \big\} \leq \pathinvsum{v} \ADVpm(\varphi_v)
 \enspace .
\eeq
\end{lemma}

\begin{proof}
The proof is by induction in the maximum distance from $v$ to a leaf.  The base case, that all of $v$'s inputs are themselves leaves is by definition of $P_v$ and $P_v^\dagger$, since then $\pathinvsum{v} = 1 + 1/\ADVpm(g_v)$.  

Let $v$ have children $c_1, \ldots, c_k$.  By \lemref{t:wsizefcompose} with $s = \vec 1$ and $S = \{ j \in [k] : \text{$c_j$ is not a leaf} \}$,   
\beq
\frac{\wsizef {P_{\varphi_v}}}{\ADVpm(\varphi_v)} \leq \frac{1}{\ADVpm(\varphi_v)} + \max_{j \in S} \max \Bigg\{ \frac{\wsizef {P_{\varphi_{c_j}}}}{\ADVpm(\varphi_{c_j})}, \frac{\wsizef {P_{\varphi_{c_j}}^\dagger}}{\ADVpm(\varphi_{c_j})} \Bigg\}
 \enspace .
\eeq
In the case $\varphi_v(x) = 1$, this follows since $P_v$ is strict, so in Eq.~\eqnref{e:wsizefcomposetrue} the sum over $\Ifree$ is zero.  In the case $\varphi_v(x) = 0$, this follows since $P_v$ is in canonical form, so in Eq.~\eqnref{e:wsizefcomposefalse}, $\norm{\ket{w'}}^2 = 1$.  

Now by induction, the right-hand side is at most $\ADVpm(\varphi_v)^{-1} + \max_{j \in S} \pathinvsum{\varphi_{c_j}} = \pathinvsum{v}$.  
\end{proof}

In particular, applying \lemref{t:unbalancedformulaevaluationtruefalsewitness} for the case $v = r$, we find 
\beq \label{e:unbalancedformulaevaluationwsize}
\wsizef {P_\varphi} \leq \pathinvsum{\varphi} \ADVpm(\varphi) = O\big(\beta^2 \ADVpm(\varphi) \big)
\eeq
since $\pathinvsum{\varphi} = O(\beta^2)$ by \lemref{t:balancegeometric}.  Combining Eqs.~\eqnref{e:unbalancedformulaevaluationnorm} and~\eqnref{e:unbalancedformulaevaluationwsize} gives 
\beq
\wsizef {P_\varphi} \, \norm{\abst(A_{G_{P_\varphi}})} = \beta^3 \, 2^{O(k_{\text{max}})} \ADVpm(\varphi)
 \enspace .
\eeq
This is $O(\ADVpm(\varphi))$; since the gate set $\cS$ is fixed and finite, $k_{\text{max}} = O(1)$.  \thmref{t:unbalancedformulaevaluation} now follows from \thmref{t:generalspanprogramalgorithmnonblackbox}.  
\end{proof}

Note that the lost constant in the theorem grows cubically in the balance parameter $\beta$ and exponentially in the maximum fan-in $k_{\text{max}}$ of a gate in $\cS$.  It is conceivable that this exponential dependence can be improved.  

For future reference, we state separately the bound used above to derive Eq.~\eqnref{e:unbalancedformulaevaluationnorm}.  

\begin{lemma} \label{t:directsumnorm}
If $P_\varphi$ is the direct-sum composition along a formula $\varphi$ of span programs $P_v$ and $P_v^\dagger$, then 
\beq
\norm{\abst(A_{G_P})} \leq 2 \max_{v \in \varphi} \max \{ \norm{\abst(A_{G_{P_v}})}, \norm{\abst(A_{G_{P_v^\dagger}})} \}
 \enspace .
\eeq
If the span programs $P_v$ are monotone, then $\norm{\abst(A_{G_P})} \leq 2 \max_v \norm{\abst(A_{G_{P_v}})}$.  
\end{lemma}

The claim for monotone span programs follows because then the dual span programs $P_v^\dagger$ are not used in $P_\varphi$.

\section{Evaluation of approximately balanced AND-OR formulas} \label{s:approxbalancedandor}

The proof of \thmref{t:approxbalancedandor} will again be a consequence of \lemref{t:wsizefcompose} and \thmref{t:generalspanprogramalgorithmnonblackbox}.  

We will use the following strict, monotone span programs for fan-in-two AND and OR gates: 

\begin{definition} \label{t:andorspanprogramdef}
For $s_1, s_2 > 0$, define span programs $P_{\AND}(s_1, s_2)$ and $P_{\OR}(s_1, s_2)$ computing $\AND_2$ and $\OR_2$, $\B^2 \rightarrow \B$, respectively, by 
\begin{align}
P_{\AND}(s_1, s_2): &&
\ket t &= \left( \begin{matrix} \alpha_1 \\ \alpha_2 \end{matrix} \right) ,\; 
&\ket{v_1} &= \left( \begin{matrix} \beta_1 \\ 0 \end{matrix} \right) ,\; &\ket{v_2} &= \left( \begin{matrix} 0 \\ \beta_2 \end{matrix} \right) \\
P_{\OR}(s_1, s_2): &&
\ket t &= \delta ,\; &\ket{v_1} &= \epsilon_1 ,\;& \ket{v_2} &= \epsilon_2 
\end{align}
Both span programs have $I_{1,1} = \{1\}$, $I_{2,1} = \{2\}$ and $\Ifree = I_{1,0} = I_{2,0} = \emptyset$.  Here the parameters $\alpha_j, \beta_j, \delta, \epsilon_j$, for $j \in [2]$, are given by 
\begin{align}
\alpha_j &= (s_j / s_p)^{1/4} & \beta_j &= 1 \\
\delta &= 1 & \epsilon_j &= (s_j / s_p)^{1/4}
 \enspace ,
\end{align}
where $s_p = s_1 + s_2$.  Let $\alpha = \sqrt{\alpha_1^2 + \alpha_2^2}$ and $\epsilon = \sqrt{\epsilon_1^2 + \epsilon_2^2}$.  
\end{definition}

Note that $\alpha, \epsilon \in (1, 2^{1/4}]$.  They are largest when $s_1 = s_2$.  

\begin{claim} \label{t:andorspanprogramwsize}
The span programs $P_{\AND}(s_1, s_2)$ and $P_{\OR}(s_1, s_2)$ satisfy: 
\begin{align}\begin{split}
\wsizexS{P_{\AND}}{x}{(\sqrt{s_1}, \sqrt{s_2})}
&=
\begin{cases}
\sqrt{s_p} & \text{if $x \in \{11, 10, 01\}$} \\
\frac{\sqrt{s_p}}{2} & \text{if $x = 00$}
\end{cases} \\
\wsizexS{P_{\OR}}{x}{(\sqrt{s_1}, \sqrt{s_2})}
&=
\begin{cases}
\sqrt{s_p} & \text{if $x \in \{00, 10, 01\}$} \\
\frac{\sqrt{s_p}}{2} & \text{if $x = 11$}
\end{cases}
\end{split}\end{align}
\end{claim}

\begin{proof}
These are calculations using \defref{t:wsizedef} for the witness size.  Letting $\sigma = (\sqrt{s_1}, \sqrt{s_2})$, $Q = P_{\AND}(s_1, s_2)$ and $R = P_{\OR}(s_1, s_2)$, we have 
\begin{align}
\wsizexS{Q}{11}{\sigma} &= \Big(\frac{\alpha_1}{\beta_1}\Big)^2 \sqrt{s_1} + \Big(\frac{\alpha_2}{\beta_2}\Big)^2 \sqrt{s_2} = \sqrt{s_p} &
\wsizexS{Q}{10}{\sigma} &= \Big( \frac{\beta_2}{\alpha_2} \Big)^2 \sqrt{s_2} = \sqrt{s_p} \\
\wsizexS{Q}{00}{\sigma} &= \left( \Big(\frac{\alpha_1}{\beta_1}\Big)^2 \frac1{\sqrt{s_1}} + \Big(\frac{\alpha_2}{\beta_2}\Big)^2 \frac1{\sqrt{s_2}} \right)^{-1} \!\! = \frac{\sqrt{s_p}}{2} &
\wsizexS{Q}{01}{\sigma} &= \Big( \frac{\beta_1}{\alpha_1} \Big)^2 \sqrt{s_1} = \sqrt{s_p} \nonumber
\intertext{and}
\wsizexS{R}{11}{\sigma} &= \delta^2 \Big( \frac{\epsilon_1^2}{\sqrt{s_1}} + \frac{\epsilon_2^2}{\sqrt{s_2}} \Big)^{-1} = \frac{\sqrt{s_p}}{2} &
\wsizexS{R}{10}{\sigma} &= \Big( \frac{\delta}{\epsilon_1} \Big)^2 \sqrt{s_1} = \sqrt{s_p} \\
\wsizexS{R}{00}{\sigma} &= \Big(\frac{\epsilon_1}{\delta}\Big)^2 \sqrt{s_1} + \Big(\frac{\epsilon_2}{\delta}\Big)^2 \sqrt{s_2} = \sqrt{s_p} &
\wsizexS{R}{01}{\sigma} &= \Big( \frac{\delta}{\epsilon_2} \Big)^2 \sqrt{s_2} = \sqrt{s_p} \nonumber
 \enspace .
\end{align}
It is not a coincidence that $\wsizexS{Q}{x}{\sigma} = \wsizexS{R}{\bar x}{\sigma}$ for all $x \in \B^2$.  This can be seen as a consequence of De Morgan's laws and span program duality---see~\cite[Lemma~4.1]{Reichardt09spanprogram}.  
\end{proof}

\begin{proof}[Proof of \thmref{t:approxbalancedandor}]
Let $\varphi$ be an AND-OR formula of size $n$, i.e., on $n$ input bits.  

First expand out the formula so that every AND gate and every OR gate has fan-in two.  This expansion can be carried out without increasing $\pathinvsum{\varphi}$ by more than a factor of~$10$: 

\begin{lemma}[{\cite[Lemma~8]{AmbainisChildsReichardtSpalekZhang07andor}}] \label{t:gateexpansion}
For any AND-OR formula $\varphi$, one can efficiently construct an equivalent AND-OR formula $\varphi'$ of the same size, such that all gates in $\varphi'$ have fan-in at most two, and $\pathinvsum{\varphi'} = O(\pathinvsum{\varphi})$.   
\end{lemma}

Therefore we may assume that $\varphi$ is a formula over fan-in-two $\AND$ and $\OR$ gates.  

Now use direct-sum composition to compose the $\AND$ and $\OR$ gates according to the formula $\varphi$, as in the proof of \thmref{t:unbalancedformulaevaluation}.  Since the span programs for $\AND$ and $\OR$ are monotone, direct-sum composition does not make use of dual span programs computing NAND or NOR.  Therefore there is no need to specify these span programs.  At a vertex $v$, set the weights $s_1$ and $s_2$ to equal the sizes of $v$'s two input subformulas.  Let $P_v$ be the span program used at vertex $v$, $P_{\varphi_v}$ be the span program thus constructed for the subformula $\varphi_v$, and $P_\varphi$ be the span program constructed computing $\varphi$.  
With this choice of weights, it follows from \claimref{t:andorspanprogramwsize} and~\cite[Theorem~4.3]{Reichardt09spanprogram} that $\wsize {P_{\varphi_v}} = \ADVpm(\varphi_v) = \ADV(\varphi_v)$.  

Notice that for all $s_1, s_2 \in [0, \infty)$, $\norm{\abst(A_{G_{P_{\AND}(s_1, s_2)}})} = O(1)$ and $\norm{\abst(A_{G_{P_{\OR}(s_1, s_2)}})} = O(1)$.  Therefore, by \lemref{t:directsumnorm}, we obtain that $\norm{\abst(A_{G_{P_\varphi}})} = O(1)$.  

Thus to apply \thmref{t:generalspanprogramalgorithmnonblackbox} we need only bound $\wsizef {P_\varphi}$.  \lemref{t:unbalancedformulaevaluationtruefalsewitness} does not apply, because for $P_{\AND}(s_1, s_2)$, an optimal witness $\ket{w'}$ to $f_{P_{\AND}}(x) = 0$ might have $\norm{\ket{w'}}^2 > 1$, as each $\alpha_j < 1$.  (\lemref{t:unbalancedformulaevaluationtruefalsewitness} would apply had we set the parameters to be $\alpha_1 = \alpha_2 = 1$, $\beta_j = (s_p / s_j)^{1/4}$, but then $\norm{A_{G_{P_{\AND}}}}$ would not necessarily be $O(1)$.)  However, analogous to \lemref{t:unbalancedformulaevaluationtruefalsewitness}, we will show: 

\begin{lemma} \label{t:approxANDORformulaevaluationtruefalsewitness}
Let $v$ be a vertex of $\varphi$.  Then 
\beq
\wsizefx {P_{\varphi_v}} x
\leq 
\begin{cases}
\pathinvsum{v} \ADV(\varphi_v) & \text{if $\varphi_v(x) = 1$} \\
2 \pathinvsum{v} \ADV(\varphi_v) - 1 & \text{if $\varphi_v(x) = 0$}
\end{cases}
\eeq
\end{lemma}

\begin{proof}
The proof is by induction in the maximum distance from $v$ to a leaf.  The base case, that $v$'s two inputs are themselves leaves is by definition of $P_v$, since then $\pathinvsum{v} = 1 + 1/\sqrt 2$.  

Let $v$ have children $c_1$ and $c_2$.  We will use \lemref{t:wsizefcompose} with $s = \vec 1$, $S = \{ j \in [2] : \text{$c_j$ is not a leaf} \}$.  

If $\varphi_v(x) = 1$, then since $P_v$ is a strict span program, i.e., $\Ifree = \emptyset$, Eq.~\eqnref{e:wsizefcomposetrue} gives 
\beq
\frac{ \wsizefx {P_{\varphi_v}} x }{ \ADV(\varphi_v) }
\leq 
\frac{1}{\ADV(\varphi_v)} + \max_{j \in S} \frac{\wsizef {P_{\varphi_{c_j}}}} {\ADV(\varphi_{c_j})}
 \enspace .
\eeq
By induction, the right-hand side is at most $1/\ADV(\varphi_v) + \max_j \pathinvsum{c_j} = \pathinvsum{v}$.  

If $\varphi_v(x) = 0$ and $g_v$ is an OR gate, then the unique witness $\ket{w'}$ for $P_v$ has $\norm{\ket{w'}} = 1$, from \defref{t:andorspanprogramdef}.  From Eq.~\eqnref{e:wsizefcomposefalse} and the induction hypothesis, 
\beq\begin{split}
\frac{ \wsizefx {P_{\varphi_v}} x }{ \ADVpm(\varphi_v) }
&\leq 
\frac{ 1 }{\ADV(\varphi_v)} + \max_{j \in S} \Big( 2 \pathinvsum{c_j} - \frac{1}{\ADV(\varphi_{c_j})} \Big) \\
&< 
2 \pathinvsum{v} - \frac{1}{\ADV(\varphi_v)}
 \enspace ,
\end{split}\eeq
as claimed.  

Therefore assume that $\varphi_v(x) = 0$ and $g_v$ is an $\AND$ gate.  Let $s_1$ and $s_2$ be the sizes of the two input subformulas to $v$, $s_p = s_1 + s_2 = \ADV(\varphi_v)^2$, and assume without loss of generality that $\varphi_{c_1}(x) = 0$.  If $\varphi_{c_2}(x) = 0$ as well, then assume without loss of generality that $2 \pathinvsum{c_1} - \frac{1}{\sqrt{s_1}} \geq 2 \pathinvsum{c_2} - \frac{1}{\sqrt{s_2}}$, so $\sigma(\bar y) \leq 2 \pathinvsum{c_1} - \frac{1}{\sqrt{s_1}}$.  Then the witness $\ket{w'}$ may be taken to be $\ket{w'} = (1/\alpha_1, 0) = \big( (s_p/s_1)^{1/4}, 0 \big)$.  From Eq.~\eqnref{e:wsizefcomposefalse}, 
\beq\begin{split}
\frac{ \wsizefx {P_{\varphi_v}} x }{ \ADVpm(\varphi_v) }
&\leq 
\frac{ \sqrt{s_p/s_1} }{\ADVpm(\varphi_v)} + \sigma(\bar y) \\
&\leq 
\frac{1}{\sqrt{s_1}} + \Big( 2 \pathinvsum{c_1} - \frac{1}{\sqrt{s_1}} \Big) \\
&<
2 \pathinvsum{v} - \frac{1}{\sqrt{s_p}}
 \enspace ,
\end{split}\eeq
as claimed.  
\end{proof}

In particular, applying \lemref{t:approxANDORformulaevaluationtruefalsewitness} for the case $v = r$, we find 
\beq
\wsizef {P_\varphi} \leq 2 \pathinvsum{\varphi} \ADV(\varphi) = 2 \pathinvsum{\varphi} \sqrt n
 \enspace .
\eeq
\thmref{t:approxbalancedandor} now follows from \thmref{t:generalspanprogramalgorithmnonblackbox}.  
\end{proof}

\section{Open problems}

In order to begin to relax the balance condition for general formulas, it seems that we need a better understanding of the canonical span programs.  For example, can the norm bound \lemref{t:abstAGPnorm} be improved?  

Although the two-sided bounded-error quantum query complexity of evaluating formulas is beginning to be understood, the zero-error quantum query complexity~\cite{BuhrmanCleveDeWolfZalka99zeroerror} appears to be more complicated.  For example, the exact and zero-error quantum query complexities for $\OR_n$ are both~$n$~\cite{BealsBuhrmanCleveMoscaWolf98}.  On the other hand, Ambainis et al.~\cite{AmbainisChildsLegallTani09witness} use the~\cite{AmbainisChildsReichardtSpalekZhang07andor} algorithm as a subroutine in the construction of a self-certifying, zero-error quantum algorithm that makes $O(\sqrt n \log^2 n)$ queries to evaluate the balanced binary AND-OR formula.  It is not known how to relax the balance requirement or extend the gate set.  

Can we develop further methods for constructing span programs with small full witness size, norm and maximum degree?  A companion paper~\cite{Reichardt09andorfaster} studies reduced tensor-product span program composition in order to complement the direct-sum composition that we have used here.  

The case of formulas over non-boolean gates may be more complicated~\cite{Reichardt09spanprogram}, but is still intriguing.

\section*{Acknowledgements}
I thank Andrew Landahl and Robert {\v S}palek for helpful discussions.  Research supported by NSERC and ARO-DTO.

\bibliographystyle{alpha-eprint}
\bibliography{andor}

\appendix

\section{Spectral gap for approximately balanced AND-OR formulas} \label{s:andorspectralgap}

\renewcommand{\NAND}{\mathop{\overline\wedge}}

It is perhaps of interest to understand why~\cite{AmbainisChildsReichardtSpalekZhang07andor} imposes the unnecessary condition that $\pathsum{\varphi} = O(n)$.  The proof in~\cite{AmbainisChildsReichardtSpalekZhang07andor} has two main cases, an eigenvalue-zero analysis of a certain graph, and a small, nonzero eigenvalue analysis of the graph.  The quantity $\pathsum{\varphi}$ appears only in the small-eigenvalue analysis, and not in the eigenvalue-zero analysis.  However,~\cite[Theorem~8.7]{Reichardt09spanprogram} states, roughly, that the small-eigenvalue analysis of a span program is unnecessary, as it follows from the eigenvalue-zero analysis.  This provides a strong indication that the small-eigenvalue analysis in~\cite{AmbainisChildsReichardtSpalekZhang07andor} is overly conservative.  However, this conclusion is not certain to be the case, since \cite[Theorem~8.7]{Reichardt09spanprogram} only shows the existence of an ``effective" spectral gap based on the eigenvalue-zero analysis, whereas~\cite{AmbainisChildsReichardtSpalekZhang07andor} in fact proves an actual spectral gap.  

In this appendix, we therefore prove a stronger version of~\cite[Lemma~5]{AmbainisChildsReichardtSpalekZhang07andor}, that does not depend on $\pathsum{\varphi}$.  This proof can be seen as an alternative, more direct way of proving \thmref{t:approxbalancedandor}, without relying on the span program framework.  It confirms that the spectral gap is of the same size, up to constants, as the effective spectral gap from~\cite[Theorem~8.7]{Reichardt09spanprogram}.  

Replacing~\cite[Lemma~5]{AmbainisChildsReichardtSpalekZhang07andor} with \lemref{t:ybounds} below gives an alternative proof of \thmref{t:approxbalancedandor}.  We use the notation from~\cite{AmbainisChildsReichardtSpalekZhang07andor}.  

\newcommand{\z}[1]{\sqrt{s_{#1}}}
\begin{lemma} \label{t:ybounds}
Let $0 < E \le 1/\big( 8 \pathinvsum{\varphi}^3 N \big){}^{1/2}$.
For vertices $v \neq r''$ in $T$, define $y_v$ by
\beq \label{eqn:y1v}
y_v = \sqrt{s_v} \cdot \frac{\pathinvsum{v}}{1 - \gamma_v E^2}
 \enspace ,
\eeq
where $\gamma_v = 4 \pathinvsum{\varphi}^2 s_v \, \pathinvsum{v}$.  Then for every vertex $v \neq r''$ in $T$, having parent $p$, either $a_v = a_p = 0$, or
\beq
\begin{array}{l r @{\ } l}
\NAND(v) = 0 \quad \Rightarrow & 0 <  (h_{pv} a_p) / a_v &\leq y_v E \\
\NAND(v) = 1 \quad \Rightarrow & 0 > a_v / (h_{pv} a_p) &\geq -y_v E \enspace .\\
\end{array}
\eeq
\end{lemma}

\begin{proof}
The proof is by induction on the height of the tree.  Recall that $h_{pv} = (s_v / s_p)^{1/4}$.  By~\cite[Eq.~(5)]{AmbainisChildsReichardtSpalekZhang07andor}, we have the equation
\beq \label{e:hh}
E a_v = h_{pv} a_p + \sum_c h_{vc} a_c
 \enspace ,
\eeq
where the sum is over the children $c$ of $v$.  By~\cite[Lemma~8]{AmbainisChildsReichardtSpalekZhang07andor}, we may assume that $v$ has at most two children.  

For a leaf $v$, $\NAND(v)=0$ and $E a_v = h_{pv} a_p$.  Thus either $a_v = a_p = 0$, or 
\beq
\frac{h_{pv} a_p}{a_ v} = E \leq y_v E
 \enspace.
\eeq
This settles the base case of the induction argument.  

The induction proceeds as follows.  First, consider the case that $a_v = 0$.  Then the induction hypothesis implies that $a_c = 0$ for every child $c$ of $v$, whether $\NAND(v)$ is $0$ or $1$.  From Eq.~\eqnref{e:hh}, then, indeed $a_p = 0$.  Assume now that $a_v \neq 0$.  Dividing Eq.~\eqnref{e:hh} by $a_v$ and simplifying, we find 
\beq \label{e:hhdividebyav}
\frac{h_{pv} a_p}{a_v} 
= E - \sum_c h_{vc}^2 \frac{a_c}{h_{vc} a_v} 
 \enspace .
\eeq
By the induction assumption, we have
\beq
\frac{-a_c}{h_{vc} a_v} \leq \begin{cases} y_c E & \text{if $\NAND(c) = 1$} \\ \frac{-1}{y_c E} & \text{if $\NAND(c) = 0$} \end{cases}
\eeq
Substituting this bound into Eq.~\eqnref{e:hhdividebyav} gives 
\beq\begin{split} \label{e:hhinequality}
\frac{h_{pv} a_p}{a_v} 
&\leq E + \sum_{c : \NAND(c) = 1} \sqrt{\frac{s_c}{s_v}} y_c E - \sum_{c : \NAND(c) = 0} \sqrt{\frac{s_c}{s_v}} \frac{1}{y_c E} \\
&= E \Bigl(1 + \sum_{c : \NAND(c) = 1} \frac{s_c}{\sqrt{s_v}} \frac{\pathinvsum{c}}{1-\gamma_c E^2} \Bigr) - \sum_{c : \NAND(c) = 0} \sqrt{\frac{s_c}{s_v}} \frac{1}{y_c E} \\
&\leq E \Bigl( \max_c \frac{1}{1 - \gamma_c E^2} \Bigr) \Bigl( 1 + \big(\max_c \pathinvsum{c}\big) \frac{1}{\sqrt{s_v}} \sum_{c : \NAND(c) = 1} s_c \Bigr) - \sum_{c : \NAND(c) = 0} \sqrt{\frac{s_c}{s_v}} \frac{1}{y_c E}
 \enspace .
\end{split}\eeq
In particular, using that $\sum_{c : \NAND(c) = 1} s_c \leq s_v$ and $\pathinvsum{v} = \frac{1}{\sqrt{s_v}} + \max_c \pathinvsum{c}$, we obtain the bound 
\beq \label{e:hhinequality1}
\frac{h_{pv} a_p}{a_v} \leq y_v E - \sum_{c : \NAND(c) = 0} \sqrt{\frac{s_c}{s_v}} \frac{1}{y_c E}
 \enspace .
\eeq
Using instead $\gamma_c E^2 \leq 1/2$ and $\max_c \pathinvsum{c} \leq \pathinvsum{\varphi}$, we obtain the bound 
\beq \label{e:hhinequality2}
\frac{h_{pv} a_p}{a_v} \leq \frac{2 \pathinvsum{\varphi} E}{\sqrt{s_v}} \Big(\sqrt{s_v} + \sum_{c : \NAND(c) = 1} s_c\Big) - \sum_{c : \NAND(c) = 0} \sqrt{\frac{s_c}{s_v}} \frac{1}{y_c E}
 \enspace .
\eeq
We will apply both Eq.~\eqnref{e:hhinequality1} and Eq.~\eqnref{e:hhinequality2} below.  

Now we consider two cases, depending on whether $\NAND(v)$ is $0$ or $1$, i.e., on whether $\{c : \NAND(c) = 0\}$ is empty or not.  

If $\NAND(v) = 0$, then all children of $v$ evaluate to $1$.  Therefore, from Eq.~\eqnref{e:hhinequality1}, $h_{pv} a_p / a_v \leq y_v E$, as claimed.  The induction hypothesis also gives from Eq.~\eqnref{e:hhdividebyav} that $h_{pv} a_p / a_v \geq E > 0$.

If $\NAND(v) = 1$, then there is a child $c^*$ with $\NAND(c^*) = 0$.  Using that $\gamma_v E^2 \leq 1/2$ and $\sqrt{s_v / s_{c^*}} y_{c^*} E \leq y_v E \leq 1/\sqrt 2$, Eq.~\eqnref{e:hhinequality1} gives $h_{pv} a_p / a_v \leq \frac 1 {\sqrt 2} - \sqrt 2 < 0$, so $a_p \neq 0$.  We wish to argue that $-a_v / (h_{pv} a_p) \leq y_v E$.  Indeed, from Eq.~\eqnref{e:hhinequality2}, letting $C = \pathinvsum{\varphi}$ and $S_1 = \sum_{c : \NAND(c) = 1} s_c$, 
\beq\begin{split}
\frac{-a_v}{h_{pv} a_p}
&\leq 
\frac{1}{ \frac{1}{\sqrt{\frac{s_v}{s_{c^*}}} y_{c^*} E} - \frac{2 C E}{\sqrt{s_v}} ( \sqrt{s_v} + S_1 ) } \\
&\leq
\frac{\sqrt{s_v} \pathinvsum{v} E}{1 - \big(\gamma_{c^*} + \sqrt{\frac{s_v}{s_{c^*}}} y_{c^*} \frac{2 C}{\sqrt{s_v}} (\sqrt{s_v} + S_1) \big)E^2}
 \enspace ,
\end{split}\eeq
where in the second step we have multiplied numerator and denominator by $\sqrt{\frac{s_v}{s_{c^*}}} y_{c^*} E$, and applied the inequality $\frac{1}{1-a}\frac{1}{1-b} \leq \frac{1}{1-(a+b)}$.  The coefficient of $E^2$ in the denominator is bounded by
\beq\begin{split}
\gamma_{c^*} + 4 C^2 (\sqrt{s_v} + S_1)
&= 4 C^2 (s_{c^*} \pathinvsum{c^*} + \sqrt{s_v} + S_1) \\
&\leq 4 C^2 (s_v \pathinvsum{c^*} + \sqrt{s_v}) \\
&\leq 4 C^2 s_v \pathinvsum{v} = \gamma_v
 \enspace ,
\end{split}\eeq
where we have used in the first inequality that $s_{c^*} + S_1 \leq s_v$.  Hence $-a_v / (h_{pv} a_p) \leq y_v E$.  
\end{proof}

\end{document}